\documentclass[aps,superscriptaddress,11pt]{revtex4}

\usepackage[utf8x]{inputenc}

\usepackage{enumerate}
\usepackage{xcolor}
\definecolor{darkred}{rgb}{0.8,0.1,0.1}
\definecolor{darkblue2}{rgb}{0.1,0.2,0.8}
\usepackage{amsmath,mathtools}
\usepackage{enumerate}
\usepackage{graphicx,epsfig,amsmath,amssymb,verbatim,color}
\usepackage{dsfont}
\usepackage{float}
\usepackage{tikz}
\newcommand{\Choi}{Choi-Jamio\l{}kowski }
\usepackage{hyperref}
\hypersetup{colorlinks=true,citecolor=blue,linkcolor=blue,filecolor=blue,urlcolor=blue,breaklinks=true}
\usepackage{cleveref}
\usepackage{graphicx}
\usepackage{xcolor}

\definecolor{darkblue}{RGB}{0,76,156}
\definecolor{darkkblue}{RGB}{0,0,153}
\definecolor{blue2}{RGB}{102,178,255}

\usepackage{url}
\newtheorem{definition}{Definition}
\newtheorem{proposition}{Proposition}
\newtheorem{lemma}[proposition]{Lemma}

\newtheorem{theorem}[proposition]{Theorem}
\newtheorem{remark}{Remark}
\newtheorem{corollary}[proposition]{Corollary}

\def\squareforqed{\hbox{\rlap{$\sqcap$}$\sqcup$}}
\def\qed{\ifmmode\squareforqed\else{\unskip\nobreak\hfil
\penalty50\hskip1em\null\nobreak\hfil\squareforqed
\parfillskip=0pt\finalhyphendemerits=0\endgraf}\fi}
\def\endenv{\ifmmode\;\else{\unskip\nobreak\hfil
\penalty50\hskip1em\null\nobreak\hfil\;
\parfillskip=0pt\finalhyphendemerits=0\endgraf}\fi}
\newenvironment{proof}{\noindent \textbf{{Proof.}~}}{\hfill $\blacksquare$}
\mathchardef\ordinarycolon\mathcode`\:
\mathcode`\:=\string"8000
\def\vcentcolon{\mathrel{\mathop\ordinarycolon}}
\begingroup \catcode`\:=\active
  \lowercase{\endgroup
  \let :\vcentcolon
  }

\makeatletter
\def\resetMathstrut@{%
    \setbox\z@\hbox{%
        \mathchardef\@tempa\mathcode`\[\relax
        \def\@tempb##1"##2##3{\the\textfont"##3\char"}%
        \expandafter\@tempb\meaning\@tempa \relax
    }%
    \ht\Mathstrutbox@\ht\z@ \dp\Mathstrutbox@\dp\z@}
\newcommand{\nc}{\newcommand}
\nc{\rnc}{\renewcommand}
\nc{\beg}{\begin{equation}}
\nc{\eeq}{{\end{equation}}}
\nc{\beqa}{\begin{eqnarray}}
\nc{\eeqa}{\end{eqnarray}}
\nc{\lbar}[1]{\overline{#1}}
\nc{\bra}[1]{\langle#1|}
\nc{\ket}[1]{|#1\rangle}
\nc{\ketbra}[2]{|#1\rangle\!\langle#2|}
\nc{\braket}[2]{\langle#1|#2\rangle}

\nc{\proj}[1]{| #1\rangle\!\langle #1 |}
\nc{\avg}[1]{\langle#1\rangle}
\nc{\Rank}{\operatorname{Rank}}
\nc{\smfrac}[2]{\mbox{$\frac{#1}{#2}$}}
\nc{\tr}{\operatorname{Tr}}
\nc{\ox}{\otimes}
\nc{\dg}{\dagger}
\nc{\dn}{\downarrow}
\nc{\cA}{{\cal A}}
\nc{\cB}{{\cal B}}
\nc{\cC}{{\cal C}}
\nc{\cD}{{\cal D}}
\nc{\cE}{{\cal E}}
\nc{\cF}{{\cal F}}
\nc{\cG}{{\cal G}}
\nc{\cH}{{\cal H}}
\nc{\cI}{{\cal I}}
\nc{\cJ}{{\cal J}}
\nc{\cK}{{\cal K}}
\nc{\cL}{{\cal L}}
\nc{\cM}{{\cal M}}
\nc{\cN}{{\cal N}}
\nc{\cO}{{\cal O}}
\nc{\cP}{{\cal P}}
\nc{\cQ}{{\cal Q}}
\nc{\cR}{{\cal R}}
\nc{\cS}{{\cal S}}
\nc{\cT}{{\cal T}}
\nc{\cV}{{\cal V}}
\nc{\cX}{{\cal X}}
\nc{\cY}{{\cal Y}}
\nc{\cZ}{{\cal Z}}
\nc{\cW}{{\cal W}}
\nc{\csupp}{{\operatorname{csupp}}}
\nc{\qsupp}{{\operatorname{qsupp}}}
\nc{\var}{{\operatorname{var}}}
\nc{\rar}{\rightarrow}
\nc{\lrar}{\longrightarrow}
\nc{\polylog}{{\operatorname{polylog}}}
\nc{\1}{{\mathds{1}}}
\nc{\wt}{{\operatorname{wt}}}
\nc{\av}[1]{{\left\langle {#1} \right\rangle}}
\nc{\supp}{{\operatorname{supp}}}

\nc{\dia}{{\diamondsuit }}

\def\di{\diamondsuit}

\def\ve{\varepsilon}

\def\x{\xi}

\nc{\SSS}{{{\mathbb S}}}
\nc{\RR}{{{\mathbb R}}}
\nc{\CC}{{{\mathbb C}}}
\nc{\FF}{{{\mathbb F}}}
\nc{\NN}{{{\mathbb N}}}
\nc{\ZZ}{{{\mathbb Z}}}
\nc{\PP}{{{\mathbb P}}}
\nc{\QQ}{{{\mathbb Q}}}
\nc{\UU}{{{\mathbb U}}}
\nc{\EE}{{{\mathbb E}}}
\nc{\id}{{\operatorname{id}}}

\nc{\CHSH}{{\operatorname{CHSH}}}

\nc{\be}{\begin{equation}}
\nc{\ee}{{\end{equation}}}
\nc{\bea}{\begin{eqnarray}}
\nc{\eea}{\end{eqnarray}}
\nc{\<}{\langle}
\rnc{\>}{\rangle}
\nc{\Hom}[2]{\mbox{Hom}(\CC^{#1},\CC^{#2})}
\nc{\rU}{\mbox{U}}

\nc{\ob}[1]{#1}

\nc{\SEP}{{\text{SEP}}}
\nc{\NS}{{\text{NS}}}
\nc{\LOCC}{{\operatorname{LOCC}}}
\nc{\PPT}{{\operatorname{PPT}}}
\nc{\EXT}{{\text{EXT}}}
\nc{\Sym}{{\operatorname{Sym}}}

\nc{\HH}{\mathbb{H}}

\nc{\ERLO}{{E_{\text{r,LO}}}}
\nc{\ERLOCC}{{E_{\text{r,LOCC}}}}
\nc{\ERPPT}{{E_{\text{r,PPT}}}}
\nc{\ERLOCCinfty}{{E^{\infty}_{\text{r,LOCC}}}}
\nc{\Aram}{{\operatorname{\sf A}}}

\rnc{\bar}{\;\rule{0pt}{9.5pt}\right|\;}

\nc{\lset}{\left\{\left.}
\nc{\rset}{\right\}}

\nc{\lsetr}{\left\{}
\nc{\rsetr}{\right.\right\}}
\nc{\barr}{\left|\rule{0pt}{9.5pt}\;}

\let\id\1

\nc{\norm}[2]{\left\lVert#1\right\rVert_{#2\!}}
\nc{\lnorm}[2]{\left\lVert#1\right\rVert_{\ell_{#2}}}

\usepackage{tikz}
\usetikzlibrary{arrows}

\makeatletter
\def\grd@save@target#1{%
  \def\grd@target{#1}}
\def\grd@save@start#1{%
  \def\grd@start{#1}}
\tikzset{
  grid with coordinates/.style={
    to path={%
      \pgfextra{%
        \edef\grd@@target{(\tikztotarget)}%
        \tikz@scan@one@point\grd@save@target\grd@@target\relax
        \edef\grd@@start{(\tikztostart)}%
        \tikz@scan@one@point\grd@save@start\grd@@start\relax
        \draw[minor help lines,magenta] (\tikztostart) grid (\tikztotarget);
        \draw[major help lines] (\tikztostart) grid (\tikztotarget);
        \grd@start
        \pgfmathsetmacro{\grd@xa}{\the\pgf@x/1cm}
        \pgfmathsetmacro{\grd@ya}{\the\pgf@y/1cm}
        \grd@target
        \pgfmathsetmacro{\grd@xb}{\the\pgf@x/1cm}
        \pgfmathsetmacro{\grd@yb}{\the\pgf@y/1cm}
        \pgfmathsetmacro{\grd@xc}{\grd@xa + \pgfkeysvalueof{/tikz/grid with coordinates/major step}}
        \pgfmathsetmacro{\grd@yc}{\grd@ya + \pgfkeysvalueof{/tikz/grid with coordinates/major step}}
        \foreach \x in {\grd@xa,\grd@xc,...,\grd@xb}
        \node[anchor=north] at (\x,\grd@ya) {\pgfmathprintnumber{\x}};
        \foreach \y in {\grd@ya,\grd@yc,...,\grd@yb}
        \node[anchor=east] at (\grd@xa,\y) {\pgfmathprintnumber{\y}};
      }
    }
  },
  minor help lines/.style={
    help lines,
    step=\pgfkeysvalueof{/tikz/grid with coordinates/minor step}
  },
  major help lines/.style={
    help lines,
    line width=\pgfkeysvalueof{/tikz/grid with coordinates/major line width},
    step=\pgfkeysvalueof{/tikz/grid with coordinates/major step}
  },
  grid with coordinates/.cd,
  minor step/.initial=.2,
  major step/.initial=1,
  major line width/.initial=2pt,
}
\makeatother

\tikzset{
  treenode/.style = {align=center, inner sep=0pt, text centered,
    font=\sffamily},
  arn_n/.style = {treenode, circle, white, font=\sffamily\bfseries, draw=black,
    fill=black, text width=1.5em},% arbre rouge noir, noeud noir
  arn_r/.style = {treenode, circle, red, draw=red, 
    text width=1.5em, very thick},% arbre rouge noir, noeud rouge
  arn_x/.style = {treenode, rectangle, draw=black,
    minimum width=0.5em, minimum height=0.5em}% arbre rouge noir, nil
}

\usepackage{pdfpages}
\makeatletter
\AtBeginDocument{\let\LS@rot\@undefined}
\makeatother
\nc{\STAB}{{{\operatorname{Stab}}}}
\nc{\PWP}{{{\operatorname{CPWP}}}}
 \nc{\SCPO}{{{\operatorname{CSPO}}}}
 \nc{\bu}{{{\textbf{u}}}}
  \nc{\bv}{{{\textbf{v}}}}
    \nc{\bw}{{{\textbf{w}}}}
        \nc{\by}{{{\textbf{y}}}}
  \nc{\sn}{{{\operatorname{sn}}}}
\allowdisplaybreaks

\begin{document}
\title{Quantifying the magic of quantum channels}
 \author{Xin Wang}
 \email{wangxinfelix@gmail.com}
\affiliation{Joint Center for Quantum Information and Computer Science, University of Maryland, College Park, Maryland 20742, USA}
 \affiliation{Institute for Quantum Computing, Baidu Research, Beijing 100193, China}

\author{Mark M. Wilde}
\email{mwilde@lsu.edu}
\affiliation{Hearne Institute for Theoretical Physics, Department of Physics and Astronomy,
Center for Computation and Technology, Louisiana State University, Baton Rouge, Louisiana 70803, USA}

\author{Yuan Su}
  \email{buptsuyuan@gmail.com}
\affiliation{Joint Center for Quantum Information and Computer Science, University of Maryland, College Park, Maryland 20742, USA}
 \affiliation{Department of Computer Science, Institute for Advanced Computer Studies, University of Maryland, College Park, USA}
\begin{abstract} 
To achieve universal quantum computation via general fault-tolerant schemes, stabilizer operations must be supplemented with other non-stabilizer quantum resources. 
Motivated by this necessity, we develop a resource theory for magic quantum channels to characterize and quantify the quantum
``magic'' or non-stabilizerness of noisy quantum circuits. For qudit quantum computing with odd dimension $d$, it is known that quantum states with non-negative Wigner function can be efficiently simulated classically. 
First, inspired by this observation, we introduce a resource theory based on completely positive-Wigner-preserving quantum operations as free operations, and we show that they can be efficiently simulated via a classical algorithm.
Second, we introduce two efficiently computable magic measures for quantum channels, called the mana and thauma of a quantum channel. As applications, we show that these measures not only provide fundamental limits on the distillable magic of quantum channels, but they also lead to lower bounds for the task of synthesizing non-Clifford gates.
Third, we propose a classical algorithm for simulating noisy quantum circuits, whose sample complexity can be quantified by the mana of a quantum channel.  We further show that this algorithm can outperform another approach for simulating noisy quantum circuits, based on channel robustness.
Finally, we explore the threshold of non-stabilizerness for basic quantum circuits under  depolarizing noise.
\end{abstract}  

\date{\today}
\maketitle 
\tableofcontents

\section{Introduction}

\subsection{Background}

One of the main obstacles to physical realizations of quantum computation is decoherence that occurs during the execution of quantum algorithms.  Fault-tolerant quantum computation (FTQC) \cite{Shor1996,Campbell2017c} provides a framework to overcome this difficulty by encoding quantum information into quantum error-correcting codes, and it allows reliable quantum computation when the physical error rate is below a certain threshold value. 

The fault-tolerant approach to quantum computation allows for a limited set of transversal, or manifestly fault-tolerant, operations, which are usually taken to be the stabilizer operations. However, the stabilizer operations alone do not enable universality because they can be simulated efficiently on a classical computer, a result known as the Gottesman-Knill theorem \cite{Gottesman1997,Aaronson2004a}. 
The addition of non-stabilizer quantum resources, such as \textit{non-stabilizer operations}, can lead to universal quantum computation~\cite{Bravyi2005}. With this perspective,  it is natural to consider the resource-theoretic approach \cite{Chitambar2018} to quantify and characterize non-stabilizer quantum resources, including both quantum states and channels.

One solution for the above scenario is to implement a non-stabilizer operation via state injection \cite{Zhou2000} of so-called ``magic states,'' which are costly to prepare via magic state distillation \cite{Bravyi2005} (see also \cite{Bravyi2012,Jones2013,Haah2017,Campbell2018,Hastings2018,Krishna2018a,Chamberland2018}). The usefulness of such magic states also motivates the resource theory of magic states \cite{Veitch2012,Mari2012,Veitch2014,Howard2016,Bravyi2016,Wang2018}, where the free operations are the stabilizer operations and the free states are the stabilizer states (abbreviated as ``Stab''). On the other hand, since a key step of fault-tolerant quantum computing is to implement non-stabilizer operations, a natural and fundamental problem is to quantify the non-stabilizerness or ``magic'' of quantum operations. As we are at the stage of Noisy Intermediate-Scale Quantum (NISQ) technology, a resource theory of magic for noisy quantum operations is desirable both to exploit the power and to identify the limitations of NISQ devices in fault-tolerant quantum computation.

\subsection{Overview of results}

In this paper, we develop a framework for the resource theory of magic quantum channels, based on qudit systems with odd prime dimension $d$.
Related work on this topic has appeared recently \cite{Seddon2019}, but the set of free operations that we take in our resource theory is larger, given by the completely positive-Wigner-preserving operations as we detail below.
We note here that $d$-level fault-tolerant quantum computation based on qudits with prime $d$ is of considerable interest for both theoretical and practical purposes \cite{Gottesman1999a,Howard2014,Campbell2014,Anwar2012,Dawkins2015}.

Our paper is structured as follows:
\begin{itemize}

\item In Section~\ref{sec: preliminary}, we first review the stabilizer formalism \cite{Gottesman1997} and the discrete Wigner function \cite{Wootters1987,Gross2006,Gross2007}. We further review various magic measures of quantum states and introduce various classes of free operations, including the stabilizer operations and beyond. 

\item In Section~\ref{sec:magic channel measure}, we introduce and characterize the completely positive-Wigner-preserving (CPWP) operations. We then introduce two efficiently computable magic measures for quantum channels. The first is the mana of quantum channels, whose state version was introduced in \cite{Veitch2014}. The second is the max-thauma of quantum channels, inspired by the magic state measure in \cite{Wang2018}. We prove several desirable properties of these two measures, including reduction to states, faithfulness, additivity for tensor products of channels, subadditivity for serial composition of channels, an amortization inequality, and monotonicity under CPWP superchannels. 

\item In Section~\ref{sec:magic gen}, we explore the ability of quantum channels to generate magic states. We first introduce the amortized magic of a quantum channel as the largest amount of magic that can be generated via a quantum channel. Furthermore, we  introduce an information-theoretic notion of the distillable magic of a quantum channel. In particular, we show that both the amortized magic and distillable magic of a quantum channel can be bounded from above by its mana and max-thauma.

\item In Section~\ref{sec: magic cost}, we apply our magic measures for quantum channels in order to evaluate the magic cost of quantum channels, and we explore further applications in quantum gate synthesis.  In particular, we show that  at least four $T$ gates are required to perfectly implement a controlled-controlled-NOT gate.

\item In Section~\ref{sec: classical simulation}, we propose a classical algorithm, inspired by \cite{PWB15}, for simulating quantum circuits, which is relevant for the broad class of noisy quantum circuits that are currently being run on NISQ devices. This algorithm has sample complexity that scales with respect to the mana of a quantum channel. We further show by concrete examples that the new algorithm can outperform a previous approach for simulating noisy quantum circuits, based on channel robustness \cite{Seddon2019}.
 \end{itemize}

\section{Preliminaries}\label{sec: preliminary}

\subsection{The stabilizer formalism}

For most known fault-tolerant schemes, the restricted set of quantum operations is the stabilizer operations, consisting of preparation and measurement in the computational basis and a restricted set of unitary operations. Here we review the basic elements of the stabilizer states and operations for systems with a dimension that is a product of odd primes. Throughout this paper, a Hilbert space implicitly has an odd dimension, and if the dimension is not
prime, it should be understood to be a tensor product of
Hilbert spaces each having odd prime dimension.

 Let $\cH_d$ denote  a Hilbert space of dimension $d$,  and let $\{\ket j\}_{j=0,\cdots,d-1}$ denote the standard computational basis. For a prime number $d$, we define the unitary boost and shift operators $X,Z\in\cL(\cH_d)$ in terms of their action on the computational basis:
\begin{align}
X\ket j &= \ket{j\oplus 1} , \\
Z\ket j &=\omega^j \ket j, \quad\omega=e^{2\pi i /d},
\end{align}
where $\oplus$ denotes addition modulo $d$.
We define the Heisenberg--Weyl operators as
\begin{align}
T_{\bu}= \tau^{-a_1a_2}Z^{a_1}X^{a_2},
\end{align}
where $\tau=e^{(d+1)\pi i/d}$, $\bu=(a_1,a_2)\in \ZZ_d\times \ZZ_d$.

For a system with composite Hilbert space $\cH_{A}\ox\cH_B$, the Heisenberg--Weyl operators are
the tensor product of the subsystem Heisenberg--Weyl operators:
\begin{align}
T_{\bu_A\oplus \bu_B} = T_{\bu_A} \ox T_{\bu_B},
\end{align} 
where $\bu_A\oplus \bu_B$ is an element of $\ZZ_{d_A}\times\ZZ_{d_A}\times\ZZ_{d_B}\times\ZZ_{d_B}$.

The Clifford operators $\cC_d$ are defined to be the set of unitary operators that map Heisenberg--Weyl operators to Heisenberg--Weyl operators under unitary conjugation up to phases:
\begin{align}
U\in \cC_d \text{ iff } \forall \bu, \exists \theta,\bu', \text{ s.t. }
UT_\bu U^{\dagger} = e^{i\theta}T_{\bu'}.
\end{align}
These operators form the Clifford group.

The pure stabilizer states can be obtained by applying Clifford operators to the state $\ket 0$:
\begin{align}
\{S_j\}=\{U\proj0U^\dagger: U\in \cC_d \}.
\end{align}

A state is defined to be a magic or non-stabilizer state if it cannot be written as a convex combination of pure stabilizer states.

\subsection{Discrete Wigner function}

\label{sec:dWF}

The discrete Wigner function \cite{Wootters1987,Gross2006,Gross2007} was used to show the existence of bound magic states~\cite{Veitch2012}. For an overview of discrete Wigner functions, we refer to \cite{Veitch2012,Veitch2014} for more details.
See also \cite{Ferrie2011} for a review of quasi-probability representations in quantum
theory, with applications to quantum information
science. 

For each point $\bu \in \ZZ_d\times \ZZ_d$ in the discrete phase space, there is a corresponding operator $A_\bu$, and the value of the discrete Wigner representation of a state $\rho$ at this point is given by
\begin{align}
W_{\rho}(\bu)\coloneqq\frac1d \tr [A_\bu\rho],
\label{eq:def-disc-Wigner-fcn}
\end{align}
where $d$ is the dimension of the Hilbert space and $\{A_\bu\}_{\bu}$ are the phase-space point operators:
\begin{equation}
A_0 =\frac1d \sum_\bu T_\bu,\qquad
A_\bu =T_\bu A_0 T_\bu^\dag.
\end{equation}
The discrete Wigner function can be defined more generally for a Hermitian operator~$X$ acting on a space of dimension $d$ via the same formula:
\begin{equation}
W_{X}(\bu)=\frac1d \tr [A_\bu X] .
\end{equation}

For the particular case of a measurement operator $E$ satisfying $0\leq E\leq \mathds{1}$, the discrete Wigner representation is defined as
 \begin{equation}\label{eq:wig of measurement}
 	W(E|\bu)\coloneqq \tr\!\big[EA_{\bu}\big],
 \end{equation}
 i.e., without the prefactor $1/d$. The reason for this will be clear in a moment and is related to the distinction between a frame and a dual frame \cite{Ferrie_2008,Ferrie_2009,PWB15}.

Some nice properties of the set $\{A_\bu\}_{\bu}$ are listed as follows:
\begin{enumerate}
\item  $A_\bu$ is Hermitian;
\item $\sum_\bu A_\bu/d=\1$;
\item $\tr [A_\bu A_{\bu'}] =d \, \delta (\bu,\bu')$;
\item $\tr [A_\bu] =1$;
\item $\rho=\sum_\bu W_{\rho}(\bu) A_\bu$;
\item $\{A_\bu\}_{\bu}=\{A_\bu^T\}_{\bu}$.
\end{enumerate}

From the second property above and the definition in \eqref{eq:def-disc-Wigner-fcn}, we conclude the following equality for a quantum state $\rho$:
\begin{equation}
\sum_{\bu} W_{\rho}(\bu) = 1. \label{eq:dWF-normalized}
\end{equation}
For this reason, the discrete Wigner function is known as a quasi-probability distribution. More generally, for a Hermitian operator $X$, we have that
\begin{equation}
\sum_{\bu} W_{X}(\bu) = \tr[X],
\end{equation}
so that for a subnormalized state $\omega$, satisfying $\omega \geq 0$ and $\tr[\omega]\leq 1 $, we have that $
\sum_{\bu} W_{\omega}(\bu) \leq  1$.

Following the convention in \eqref{eq:wig of measurement} for measurement operators, we find the following for a positive operator-valued measure (POVM) $\{E^x\}_x$ (satisfying $E^x\geq 0 \ \forall x$ and $\sum_x E^x= \mathds{1}$):
\begin{equation}
\sum_{x} W(E^x|\bu) = 1 \label{eq:dWF-normalized-measurement-ops},
\end{equation}
so that the quasi-probability interpretation is retained for a POVM. That is, $W(E^x|\bu)$ can be interpreted as the conditional quasi-probability of obtaining outcome $x$ given input $\bu$.

We can quantify the amount of negativity in the discrete Wigner function of a state $\rho$ via the sum negativity, which is equal to the absolute sum of the negative elements of the Wigner function~\cite{Veitch2014}:
\begin{align}
\sn(\rho)\coloneqq \sum_{\bu:W_{\rho}(\bu)<0}|W_{\rho}(\bu)|=\frac{1}{2} \left(\sum_{\bu} |W_{\rho}(\bu)|- W_{\rho}(\bu) \right)
 = \frac{1}{2} \left(\sum_{\bu} |W_{\rho}(\bu)|\right) - \frac{1}{2}.
\end{align}
By definition, we find that $\sn(\rho)\geq 0$.
The \emph{mana} of a state $\rho$ is defined as \cite{Veitch2014}
\begin{align}
\cM(\rho)\coloneqq  \log \left(\sum_{\bu} |W_{\rho}(\bu)|\right) = \log (2 \cdot \sn (\rho)+1) \geq 0.
\label{eq:mana-state}
\end{align}
We define the mana more generally, as in \cite{Wang2018}, for a positive semi-definite operator $X$ via the formula
\begin{align}
\cM(X)  \coloneqq \log \!\left(\sum_{\bu} |W_{X}(\bu)|\right) 
 = \log \!\left( 2 \left[\sum_{\bu : W_{X}(\bu)<0} |W_{X}(\bu)|\right] + \tr[X] \right).
 \label{eq:mana-PSD-op}
\end{align}

We denote the set of quantum states with a non-negative Wigner function by $\cW_+$ (Wigner polytope), i.e.,
\begin{align}
\cW_+\coloneqq\{\rho: \forall \bu, W_{\rho}(\bu)\ge 0, \rho \ge 0, \tr [\rho]=1\}.
\end{align}
It is known that quantum states with non-negative Wigner function are classically simulable and thus are useless in magic state distillation \cite{Veitch2012}, which can be seen as the analog of states with positive partial transpose (PPT) in entanglement distillation~\cite{Peres1996,Horodecki1998}.

Motivated by the Rains bound \cite{Rains2001} and its variants \cite{Wang2016,Wang2017e,Fang2017,Tomamichel2015a,Tomamichel2016,Wang2016a,Wang2017d} in entanglement theory, the set of sub-normalized states with non-positive mana was introduced as follows~\cite{Wang2018} to explore the resource theory of magic states:
\begin{align}
	\cW\coloneqq\left\{\sigma: \sum_{\bu} |W_{\sigma}(\bu)|\le 1, \sigma \ge 0\right\}
	= \left\{\sigma: \cM(\sigma)\le 0, \sigma \ge 0\right\}.
\end{align}
It follows from definitions and the triangle inequality that $\tr[\sigma] \leq 1$ if $\sigma \in \cW$ (alternatively one can conclude this by inspecting the right-hand side of \eqref{eq:mana-PSD-op}).

Furthermore, we define $\widehat \cW_+$ to be the set of Hermitian operators with non-negative Wigner function:
\begin{align}
\widehat \cW_+\coloneqq\{V: \forall \bu, W_{V}(\bu)\ge 0\}.
\end{align}

The Wigner trace norm and Wigner spectral norm of an Hermitian operator $V$ are defined as follows, respectively:
\begin{align}
\|V\|_{W,1}&\coloneqq\sum_\bu |W_{V}(\bu)| =  \sum_\bu |\tr [A_\bu V]/d|,\\
\|V\|_{W,\infty}&\coloneqq d \max_{\bu} |W_{V}(\bu)| = \max_\bu |\tr [A_\bu V]|. \label{eq: inf norm wigner}
\end{align}
The Wigner trace and spectral norms are dual to each other in the following sense:
\begin{align}
\|V\|_{W,1} & \coloneqq\max_C \{|\tr [VC]|: \|C\|_{W,\infty}\le 1 \}, \\
\|V\|_{W,\infty} & \coloneqq\max_C \{|\tr [VC]|: \|C\|_{W,1}\le 1\},
\end{align}
with $C$ ranging over Hermitian operators within the same space.

\subsection{Stabilizer channels and beyond}

A stabilizer operation (SO) consists of the following types of quantum operations: Clifford operations, tensoring in stabilizer states, partial trace, measurements in the computational basis, and post-processing conditioned on these measurement results. Any quantum protocol composed of these quantum operations can be written in terms of the following Stinespring dilation representation: $\cE(\rho)=\tr_E[U(\rho\ox\rho_E)U^\dagger]$, where $U$ is a Clifford unitary and the ancilla $\rho_E$ is a stabilizer state.

The authors of~\cite{Ahmadi2017} generalized the set of stabilizer operations to stabilizer-preserving operations, which are those  that transform stabilizer states to stabilizer states and which form the largest set of physical operations that can be considered free for the resource theory of non-stabilizerness. More recently, Ref.~\cite{Seddon2019} introduced the completely stabilizer-preserving operations (CSPO); i.e., a quantum operation $\Pi$ is called completely stabilizer-preserving  if for any reference system $R$,
\begin{align}
\forall \rho_{RA} \in {\STAB}, \quad (\operatorname{id}_R\ox \Pi_{A\to B})(\rho_{RA}) \in {\STAB}.
\end{align}

\subsection{Magic measures of quantum states}

We review some of the magic measures of quantum states in Table~\ref{table:zoo}.
\begin{table}
		%\centering
		\begin{tabular}{l|c|c}
			\hline
			Measures & Acronym & Definition  \\
			\hline
			Mana \cite{Veitch2014} & $\cM(\rho)$  &
			\footnotesize $\log \sum_\bu |\tr A_\bu\rho|/d$   \\
			Robustness of magic \cite{Howard2016}& $\cR(\rho)$  &
			\footnotesize $\inf\{2r+1: \frac{\rho+r\sigma}{1+r}=\tau, \ \ \sigma,\tau\in \STAB\}$   \\		
			Relative entropy of magic \cite{Veitch2014} & $R_\cM(\rho)$  &
			\footnotesize $\inf_{\sigma\in \STAB}D(\rho\|\sigma)$   \\
			Regularized relative entropy of magic \cite{Veitch2014} & $R^{\infty}_\cM(\rho)$  &
			\footnotesize $\lim_{n\to \infty} R^{\infty}_\cM(\rho^{\ox n})/n$   \\
			Max-thauma \cite{Wang2018} & $\theta_{\max}(\rho)$  &
			\footnotesize $\inf_{\sigma\in \cW}D_{\max}(\rho\|\sigma)$   \\
			Thauma \cite{Wang2018} & $\theta(\rho)$  &
			\footnotesize $\inf_{\sigma\in \cW}D(\rho\|\sigma)$   \\
			Regularized thauma \cite{Wang2018} & $\theta^{\infty}(\rho)$  &
			\footnotesize $\lim_{n\to \infty} \theta(\rho^{\ox n})/n$   \\
			Min-thauma \cite{Wang2018} & $\theta_{\min}(\rho)$  &
			\footnotesize $\inf_{\sigma\in \cW}D_0(\rho\|\sigma)$ \\
			\hline
		\end{tabular}
		\caption{Partial zoo of magic measures.}
		\label{table:zoo}
	\end{table}
In particular, the max-thauma of a quantum state $\rho$ is defined as follows \cite{Wang2018}:
	\begin{align}
	\theta_{\max}(\rho) & \coloneqq\min_{\sigma\in \cW} D_{\max}(\rho \| \sigma) 
	 \coloneqq \min_{\sigma\in \cW} \left[\min \{\lambda: \rho \leq 2^{\lambda} \sigma\}\right]   \\
	&=\log_2 \min\left\{ \Vert V \Vert_{W,1} 
	:  \rho \le V \right\},  	\label{eq:SDP theta min}
	\end{align}
	where the max-relative entropy $D_{\max}(\rho \| \sigma)$ was defined in \cite{Datta2009}.
%%%%%%%%%%%%%%%%%%%%%%%%%%%%%%
%%%%%%%%%%%%%%%%%%%%%%%%%%

\section{Quantifying the non-stabilizerness of a quantum channel}

\label{sec:magic channel measure}

\subsection{Completely Positive-Wigner-Preserving operations}

\label{sec:PWP channel}

A quantum circuit consisting of an initial quantum state, unitary evolutions, and measurements, each having non-negative Wigner functions, can be classically simulated \cite{PWB15}. It is thus natural to consider free operations to be those that completely preserve the positivity of the Wigner function. Indeed, any such quantum operations are proved to be efficiently simulated via classical algorithms in Section~\ref{sec: classical simulation} and thus become reasonable free operations for the resource theory of magic.
\begin{definition}[Completely PWP operation]
A Hermiticity-preserving linear  map $\Pi$ is called completely positive Wigner preserving (CPWP) if for any system $R$ with odd dimension, the following holds
\begin{align}
\forall \rho_{RA} \in {\cW_+}, \quad (\operatorname{id}_R\ox \Pi_{A\to B})(\rho_{RA}) \in {\cW_+}\ .
\end{align}
\end{definition}

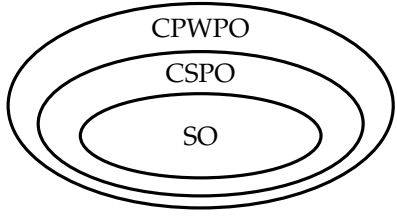
\begin{figure}
    \centering
    \begin{tikzpicture}[scale=0.8]
  %% EB
  \draw[very thick] (0,-0.1) ellipse (2.7cm and 1.2cm);
  %% CQ
  \draw[very thick] (0,-0.3) ellipse (2cm and 0.7cm);
  %% all
  \draw[very thick] (0,0.2) ellipse (3.2cm and 1.7cm);
  \node[] at (0,-0.3) {SO};
  \node[] at (0,0.75) {CSPO};
  \node[] at (0,1.5) {CPWPO};
\end{tikzpicture}
\caption{Relationship between stabilizer operations, completely stabilizer-preserving operations, and completely PWP operations.}
\label{fig: CQ EB 1}
\end{figure}
Figure~\ref{fig: CQ EB 1} depicts the relationship between stabilizer operations, completely stabilizer-preserving operations, and completely PWP operations.

We now recall the definition of the discrete Wigner function of a quantum channel from \cite{Mari2012}, which is strongly related to the Wigner function of a quantum channel as defined in \cite[Eq.~(95)]{BRS12}.
\begin{definition}[Discrete Wigner function of a quantum channel]
\label{def:dWf-channel}
Given a quantum channel $\cN_{A\to B}$, its discrete Wigner function is defined as 
\begin{align}
\cW_{\cN}(\bv|\bu) & \coloneqq \frac{1}{d_B} \tr [((A_A^{\bu})^T\ox A_B^{\bv})J^{\cN}_{AB}] \\
& = \frac{1}{d_B} \tr [A_B^{\bv}\cN(A_A^{\bu})] .
\label{eq:dWF-channel}
\end{align}
Here $J^{\cN}_{AB} =\sum_{ij} \ketbra{i}{j}_A \ox \cN(\ketbra{i}{j}_{A'})$ denotes the  \Choi matrix \cite{Choi1975,Jamiokowski1972} of the channel $\cN$, where $\{\ket{i}_A\}_i$ and $\{\ket{i}_{A'}\}_i$ are orthonormal bases on isomorphic Hilbert spaces $\cH_A$ and $\cH_{A'}$, respectively. More generally, the discrete Wigner function of a Hermiticity-preserving linear map $\cP_{A\to B}$ can be defined using the same formula in \eqref{eq:dWF-channel}, by substituting $\cN$ therein with $\cP$.
\end{definition}

From the definition above and the properties recalled in Section~\ref{sec:dWF}, it follows for a quantum channel $\cN_{A \to B}$ that
\begin{equation}
\sum_{\bv} \cW_{\cN}(\bv|\bu) = 1   \quad \forall \bu,
\label{eq:normalized-dWF-channel}
\end{equation}
because
\begin{align}
\sum_{\bv} \cW_{\cN}(\bv|\bu) & = \sum_{\bv} \tr [A_B^{\bv}\cN(A_A^{\bu})]/d_B  =  \tr \left[ \left(\sum_{\bv}\frac{A_B^{\bv}}{ d_B} \right)\cN(A_A^{\bu})\right] \\ 
& =  \tr [I_B \cN(A_A^{\bu})] = \tr [\cN(A_A^{\bu})]   = \tr [A_A^{\bu}] = 1,
\end{align}
where the penultimate equality follows from the fact that $\cN$ is trace preserving (in fact here we did not require complete positivity or even linearity). Due to the normalization in \eqref{eq:normalized-dWF-channel}, $\cW_{\cN}(\bv|\bu)$ can be interpreted as a conditional quasi-probability distribution.

Furthermore, the discrete Wigner function of a channel allows one to determine the output Wigner function from the input Wigner function by propagating the quasi-probability distribitions, just as one does in the classical case. When there is no reference system, such a statement was proved in \cite{Mari2012}. Here we slightly extend this result to the case with a reference system in the following lemma.

\begin{lemma}
\label{lem:prop-dWF}
For an input state $\rho_{AR}$ and a quantum channel $\cN_{A\to B}$ with respective Wigner functions $W_{\rho_{AR}}(\bu,\by)$ and $W_{\cN}(\bv|\bu)$, the Wigner function $W_{\cN(\rho_{AR})}(\bv,\by)$ of the output state $\cN_{A \to B}(\rho_{AR})$ is given by
\begin{equation}
W_{\cN_{A \to B}(\rho_{AR})}(\bv,\by) = \sum_{\bu} W_{\cN}(\bv|\bu) \ W_{\rho_{AR}}(\bu,\by).
\end{equation}
\end{lemma}
\begin{proof}
The proof is straightforward:
\begin{align}
W_{\cN(\rho_{AR})}(\bv,\by) & = \frac{1}{d_B d_R} \tr [(A_B^\bv \ox A_R^{\by}) \cN_{A \to B}(\rho_{AR})] \label{eq:dWF-state-from-channel-1}\\
& = \frac{1}{d_B d_R }\sum_{\bu,\bw}\tr [(A_B^\bv \ox A_R^{\by}) \cN_{A \to B}(W_{\rho_{AR}}(\bu,\bw)A_A^\bu\ox A_R^{\bw})]\\
& = \frac{1}{d_B d_R }\sum_{\bu,\bw} W_{\rho_{AR}}(\bu,\bw) \tr [A_B^\bv \cN_{A \to B}(A_A^\bu) \ox A_R^{\by} A_R^{\bw})]\\
& = \frac{1}{d_B d_R } \sum_{\bu,\bw} W_{\rho_{AR}}(\bu,\bw) \tr [A_B^\bv \cN_{A \to B}(A_A^\bu)] d_R\ \delta(\by,\bw)\\
& = \sum_{\bu} W_{\cN}(\bv|\bu) \ W_{\rho_{AR}}(\bu,\by).
\end{align}
All steps follow from definitions and the properties of the phase-space point operators recalled in Section~\ref{sec:dWF}.
In particular, we made use of the fact that $\rho_{AR} = \sum_{\bv,\bw} W_{\rho_{AR}}(\bv,\bw)A_A^\bv\ox A_R^{\bw}$ in the second equality.
\end{proof}

\bigskip
\begin{theorem}\label{th: PWP}
The following statements about CPWP operations are equivalent:
\begin{enumerate}
\item The quantum channel $\cN$ is CPWP;
\item The discrete Wigner function of the \Choi matrix $J_\cN$ is non-negative;
\item $W_{\cN}(\bv|\bu)$ is non-negative for all $\bu$ and $\bv$ (i.e., $W_{\cN}(\bv|\bu)$ is a conditional probability distribution or classical channel).
\end{enumerate}
\end{theorem}
\begin{proof}
$1\to 2$: Let us first apply the (stabilizer) qudit controlled-NOT gate $\rm{CNOT_d}$ to the stabilizer state $\ket + \ox \ket 0$ to prepare the maximally entangled state $\Phi_d\in \cW_+$.
Since $\cN$ completely preserves the positivity of the Wigner function, it follows that
\begin{align}
J_\cN/d=(\operatorname{id}_R\ox \cN)(\Phi_d) \in \cW_+.
\end{align}

$2\to 3$: We find that
\begin{align}
W_{\cN}(\bv|\bu)=\tr [J_{\cN}((A_A^{\bu})^T\ox A_B^{\bv})]/d_B=  \tr [J_{\cN}(A_A^{\bu'}\ox A_B^{\bv})]/d_B \ge 0.
\label{eq:dWF-channel-non-neg}
\end{align}
In the last inequality, we note that $A_A^{\bu'}=(A_A^{\bu})^T$ and we can always find such $\bu'$ since $\{A_A^\bu\}_{\bu}=\{(A_A^\bu)^T\}_{\bu}$. The fact that $W_{\cN}(\bv|\bu)$ is a conditional probability distribution follows from the inequality in \eqref{eq:dWF-channel-non-neg} and the constraint in \eqref{eq:normalized-dWF-channel}.

$3\to 1$:
If the channel $\cN$ has a non-negative Wigner function, then for an input state $\rho_{AR}$ such that $\rho_{AR}\in\cW_+$, it follows from Lemma~\ref{lem:prop-dWF} that
\begin{align}
W_{\cN(\rho_{AR})}(\bv,\by) = \sum_{\bu} W_{\cN}(\bv|\bu) \ W_{\rho_{AR}}(\bu,\by)   
 \ge 0,
\end{align}
concluding the proof.
\end{proof}

\bigskip
We remark here that the equivalence between $2$ and $3$ above was proved in \cite{Mari2012}, and our contribution is to show the equivalence between 2, 3, and the completely positive Wigner preserving property, which considers information processing in the presence of reference systems.

\subsection{Quantum (CPWP)\ superchannels}

A superchannel $\Xi_{\left(  A\rightarrow B\right)  \rightarrow\left(
C\rightarrow D\right)  }$\ is a quantum-physical evolution of a quantum
channel $\mathcal{N}_{A\rightarrow B}$ \cite{CDP08,Gour2018}, which leads to an output
channel $\mathcal{K}_{C\rightarrow D}$ as%
\begin{equation}
\mathcal{K}_{C\rightarrow D}=\Xi_{\left(  A\rightarrow B\right)
\rightarrow\left(  C\rightarrow D\right)  }(\mathcal{N}_{A\rightarrow
B}).\label{eq:output-of-superch}%
\end{equation}
The output channel $\mathcal{K}_{C\rightarrow D}$ taking system $C$ to system $D$ can be
denoted by $\Xi(\mathcal{N})$ for short. The key property of a quantum
superchannel is that the output map%
\begin{equation}
\left(  \operatorname{id}_{R}\otimes\Xi_{\left(  A\rightarrow B\right)
\rightarrow\left(  C\rightarrow D\right)  }\right)  (\mathcal{N}%
_{RA\rightarrow RB})
\end{equation}
is a legitimate quantum channel for all input bipartite channels
$\mathcal{N}_{RA\rightarrow RB}$, where the reference system $R$ is
arbitrary and $\operatorname{id}_{R}$ denotes the identity
superchannel. A superchannel $\Xi_{\left(  A\rightarrow B\right)
\rightarrow\left(  C\rightarrow D\right)  }$\ has a physical realization in
terms of a pre-processing channel $\mathcal{E}_{C\rightarrow AM}$ and a
post-processing channel $\mathcal{D}_{BM\rightarrow D}$ \cite{CDP08,Gour2018}, so that%
\begin{equation}
\Xi_{\left(  A\rightarrow B\right)  \rightarrow\left(  C\rightarrow
D\right)  }(\mathcal{N}_{A\rightarrow B})=\mathcal{D}_{BM\rightarrow D}%
\circ\mathcal{N}_{A\rightarrow B}\circ\mathcal{E}_{C\rightarrow AM}.
\end{equation}
The superchannel $\Xi_{\left(  A\rightarrow B\right)  \rightarrow\left(
C\rightarrow D\right)  }$ is in one-to-one correspondence with a bipartite
channel $\mathcal{P}_{CB\rightarrow AD}$, defined as%
\begin{equation}
\mathcal{P}_{CB\rightarrow AD}:=\mathcal{D}_{BM\rightarrow D}\circ
\mathcal{E}_{C\rightarrow AM}.
\end{equation}
Related to this, an arbitrary bipartite channel $\mathcal{P}_{CB\rightarrow
AD}^{\prime}$ is in one-to-one correspondence with a superchannel
$\Xi_{\left(  A\rightarrow B\right)  \rightarrow\left(  C\rightarrow
D\right)  }^{\prime}$ as long as it obeys the following non-signaling
constraint \cite[Theorem~4]{PHHH06}:%
\begin{equation}
\operatorname{Tr}_{D}\circ\mathcal{P}_{CB\rightarrow AD}^{\prime
}=\operatorname{Tr}_{D}\circ\mathcal{P}_{CB\rightarrow AD}^{\prime}%
\circ\mathcal{R}_{B}^{\pi},\label{eq:non-sig-superch-biparti}%
\end{equation}
where $\mathcal{R}_{B}^{\pi}$ is a replacer channel, defined as $\mathcal{R}%
_{B}^{\pi}(\omega_{B})=\operatorname{Tr}[\omega_{B}]\pi_{B}$ with $\pi_{B}$
the maximally mixed state. That is, the non-signaling constraint implies that
a trace out of system $D$ has the effect of tracing and replacing system $B$,
thus preventing $B$ from signaling to $A$.

The Choi operator of a quantum superchannel $\Xi_{\left(  A\rightarrow
B\right)  \rightarrow\left(  C\rightarrow D\right)  }$ is given by the Choi
operator of its corresponding bipartite channel $\mathcal{P}_{CB\rightarrow
AD}$ \cite{CDP08,Gour2018}:%
\begin{equation}
J_{CBAD}^{\Xi}:=\sum_{i,j,i^{\prime},j^{\prime}}|i\rangle\langle
j|_{C}\otimes|i^{\prime}\rangle\langle j^{\prime}|_{B}\otimes\mathcal{P}%
_{C^{\prime}B^{\prime}\rightarrow AD}(|i\rangle\langle j|_{C^{\prime}}%
\otimes|i^{\prime}\rangle\langle j^{\prime}|_{B^{\prime}}),
\end{equation}
where systems $B^{\prime}$ and $C^{\prime}$ are isomorphic to $B$ and $C$,
respectively. It obeys the following constraints:%
\begin{align}
J_{CBAD}^{\Xi}  & \geq0,\\
\operatorname{Tr}_{AD}[J_{CBAD}^{\Xi}]  & =I_{CB},\\
\operatorname{Tr}_{D}[J_{CBAD}^{\Xi}]  & =\operatorname{Tr}_{DB}%
[J_{CBAD}^{\Xi}]\otimes\pi_{B},
\end{align}
which correspond respectively to complete positivity of the corresponding
bipartite channel $\mathcal{P}_{CB\rightarrow AD}$, trace preservation of the
bipartite channel $\mathcal{P}_{CB\rightarrow AD}$, and the
$B\not \rightarrow A$ non-signaling constraint. Conversely, any operator
obeying the three constraints above is a bipartite channel corresponding to a
superchannel. One can employ the following propagation rule \cite{CDP08,Gour2018} to
determine the Choi operator of the output channel in
\eqref{eq:output-of-superch}:%
\begin{equation}
J_{CD}^{\mathcal{K}}=\operatorname{Tr}_{AB}\!\left[  \left(  \left(
J_{AB}^{\mathcal{N}}\right)  ^{T}\otimes I_{CD}\right)  J_{CBAD}^{\Xi
}\right]  ,\label{eq:prop-rule-superch}%
\end{equation}
where the superscript $T$ denotes the transpose operation.

By employing Definition~\ref{def:dWf-channel}, we define the discrete Wigner function of a
quantum superchannel~$\Xi_{(A\to B)\to (C\to D)}$, and we do so by means of its corresponding
bipartite channel $\mathcal{P}_{CB\rightarrow AD}$. That is, since
$\mathcal{P}_{CB\rightarrow AD}$ is a channel, it has a discrete Wigner
function
\begin{equation}
W_{\Xi}(\boldsymbol{u}_{A},\boldsymbol{v}_{D}|\boldsymbol{u}_{C}%
,\boldsymbol{v}_{B})=\frac{1}{d_{A}d_{D}}\operatorname{Tr}[\left(
A_{A}^{\boldsymbol{u}_{A}}\otimes A_{D}^{\boldsymbol{v}_{D}}\right)
\mathcal{P}_{CB\rightarrow AD}\left(  A_{C}^{\boldsymbol{u}_{C}}\otimes
A_{B}^{\boldsymbol{v}_{B}}\right)  ],
\label{eq:dWf-superch}
\end{equation}
where we use the subscript $\Xi$ to indicate its assocation with the
superchannel $\Xi$ and the choice of letters $\boldsymbol{u}$ and
$\boldsymbol{v}$ are made in the above way because, in what follows, we will
link up the discrete Wigner function $W_{\mathcal{N}}(\boldsymbol{v}%
_{B}|\boldsymbol{u}_{A})$ of a quantum channel $\mathcal{N}_{A\rightarrow B}$
with $W_{\Xi}$ and the notation given above is more convenient for doing
so. In addition to obeying the following property%
\begin{equation}
\sum_{\boldsymbol{u}_{A},\boldsymbol{v}_{D}}W_{\Xi}(\boldsymbol{u}%
_{A},\boldsymbol{v}_{D}|\boldsymbol{u}_{C},\boldsymbol{v}_{B})=1,
\end{equation}
so that $W_{\Xi}(\boldsymbol{u}_{A},\boldsymbol{v}_{D}|\boldsymbol{u}%
_{C},\boldsymbol{v}_{B})$ is a conditional quasi-probability distribution,
there is an extra constraint imposed on $W_{\Xi}(\boldsymbol{u}%
_{A},\boldsymbol{v}_{D}|\boldsymbol{u}_{C},\boldsymbol{v}_{B})$ related to the
non-signaling constraint $B\not \rightarrow A$ in \eqref{eq:non-sig-superch-biparti}. To see this, let%
\begin{equation}
W_{\Xi}(\boldsymbol{u}_{A}|\boldsymbol{u}_{C},\boldsymbol{v}_{B}%
):=\sum_{\boldsymbol{v}_{D}}W_{\Xi}(\boldsymbol{u}_{A},\boldsymbol{v}%
_{D}|\boldsymbol{u}_{C},\boldsymbol{v}_{B}).
\end{equation}
By employing the non-signaling constraint in
\eqref{eq:non-sig-superch-biparti}\ and properties of the phase-space point
operators, it is straightforward to conclude that the non-signaling constraint
$B\not \rightarrow A$ is equivalent to the following condition on the discrete
Wigner function $W_{\Xi}(\boldsymbol{u}_{A},\boldsymbol{v}_{D}%
|\boldsymbol{u}_{C},\boldsymbol{v}_{B})$:%
\begin{equation}
W_{\Xi}(\boldsymbol{u}_{A}|\boldsymbol{u}_{C},\boldsymbol{v}_{B}%
)=W_{\Xi}(\boldsymbol{u}_{A}|\boldsymbol{u}_{C},\boldsymbol{v}_{B}^{\prime
})\qquad\forall\boldsymbol{v}_{B},\boldsymbol{v}_{B}^{\prime}%
,\label{eq:non-sig-dWf-superch}%
\end{equation}
so that we can write%
\begin{equation}
W_{\Xi}(\boldsymbol{u}_{A}|\boldsymbol{u}_{C}):=W_{\Xi}(\boldsymbol{u}%
_{A}|\boldsymbol{u}_{C},\boldsymbol{v}_{B}).
\end{equation}
This can be interpreted as indicating that the output phase-space point
$\boldsymbol{u}_{A}$ is independent of $\boldsymbol{v}_{B}$ if system $D$ is
not available (i.e., has been marginalized). We note here that the conditions
in \eqref{eq:non-sig-dWf-superch} represent a direct generalization of
non-signaling constraints for classical probability distributions to
quasi-probability distributions. Furthermore, we also observe that the super-quasi-probability distribution in \eqref{eq:dWf-superch}  represents a generalization of the classical superchannels discussed in \cite{Matthews12,WM16}.

By employing the propagation rule in \eqref{eq:prop-rule-superch}\ and a
sequence of steps similar to those given in the proof of Lemma~\ref{lem:prop-dWF}, we conclude
that the discrete Wigner function of the output channel $\mathcal{K}%
_{C\rightarrow D}=\Xi_{\left(  A\rightarrow B\right)  \rightarrow\left(
C\rightarrow D\right)  }(\mathcal{N}_{A\rightarrow B})$ is given by%
\begin{equation}
W_{\Xi(\mathcal{N})}(\boldsymbol{v}_{D}|\boldsymbol{u}_{C})=\sum
_{\boldsymbol{u}_{A},\boldsymbol{v}_{B}}W_{\Xi}(\boldsymbol{u}%
_{A},\boldsymbol{v}_{D}|\boldsymbol{u}_{C},\boldsymbol{v}_{B})W_{\mathcal{N}%
}(\boldsymbol{v}_{B}|\boldsymbol{u}_{A}).\label{eq:output-dWf-superch-dWf},
\end{equation}
again generalizing the fully classical case from \cite{Matthews12,WM16}.

We now define CPWP superchannels as free superchannels that extend the notion of CPWP channels:
\begin{definition}
[CPWP\ superchannel] \label{def:CPWP-superch} A superchannel $\Xi_{\left(  A\rightarrow B\right)
\rightarrow\left(  C\rightarrow D\right)  }$ is completely CPWP\ preserving
(CPWP\ superchannel for short) if, for all CPWP\ channels $\mathcal{N}%
_{RA\rightarrow RB}$, the output channel $\Xi_{\left(  A\rightarrow
B\right)  \rightarrow\left(  C\rightarrow D\right)  }(\mathcal{N}%
_{RA\rightarrow RB})$ is CPWP, where $R$ is an arbitrary reference system.
\end{definition}

We then have the following theorem as a generalization of Theorem~\ref{th: PWP} (its proof
is very similar and so we omit it):

\begin{theorem}
\label{thm:CPWP-superchs-equiv}
The following statements about CPWP superchannels are equivalent:
\begin{enumerate}
\item The quantum
superchannel $\Xi$ is CPWP;
\item The discrete Wigner function of the Choi matrix
$J_{CBAD}^{\Xi}$ is non-negative;
\item The discrete Wigner function $W_{\Xi
}(\boldsymbol{u}_{A},\boldsymbol{v}_{D}|\boldsymbol{u}_{C},\boldsymbol{v}%
_{B})$ is non-negative for all $\boldsymbol{u}_{A}$, $\boldsymbol{v}_{D}$,
$\boldsymbol{u}_{C}$, and $\boldsymbol{v}_{B}$ (i.e., $W_{\Xi
}(\boldsymbol{u}_{A},\boldsymbol{v}_{D}|\boldsymbol{u}_{C},\boldsymbol{v}%
_{B})$ is a conditional probability distribution or classical bipartite
channel with a non-signaling constraint).
\end{enumerate}
\end{theorem}

An interesting consequence of the third part of the above theorem is that every CPWP superchannel has a non-unique realization in terms of pre- and post-processing CPWP channels. This follows from the fact that every non-signaling classical bipartite channel can be realized in terms of pre- and post-processing classical channels (see the discussion surrounding \cite[Eq.~(7)]{Matthews12}), and these  pre- and post-processing classical channels can be identified as the discrete Wigner functions of pre- and post-processing CPWP channels.

 \subsection{Logarithmic negativity (mana) of a quantum channel}
 
To quantify the magic of quantum channels, we introduce the mana (or logarithmic negativity) of a quantum channel $\cN_{A\to B}$:
\begin{definition}[Mana of a quantum channel]
The mana of a quantum channel $\cN_{A\to B}$ is defined as
\begin{align}
\cM(\cN_{A\to B}) & \coloneqq \log \max_{\bu} \|\cN_{A \to B}(A_A^\bu)\|_{W,1} \label{eq:channel-mana}\\
& = \log \max_{\bu} \sum_{\bv}\frac{1}{d_B} \left|\tr [A_B^\bv \cN_{A \to B}(A_A^\bu)]\right|\\
& = \log \max_{\bu} \sum_{\bv}\frac{1}{d_B} \left|\tr [(A_A^\bu \ox A_B^\bv) J^\cN_{AB}]\right|
\label{eq:choi-mana}\\
& = \log \max_{\bu} \sum_{\bv}  |W_{\cN}(\bv|\bu)| .
\label{eq:dWf-mana-def}
\end{align}
More generally, we define the mana of a Hermiticity-preserving linear map $\cP_{A \to B}$ via the same formula above, but substituting $\cN$ with $\cP$.
\end{definition}

In the following, we are going to show that the mana of a quantum channel has many desirable properties, such as
\begin{enumerate}
\item Reduction to states: $\cM(\cN) = \cM(\sigma)$ when the channel $\cN$ is a replacer channel, acting as $\cN(\rho) = \tr[\rho]\sigma$ for an arbitrary input state $\rho$, with $\sigma$ a state.
\item Additivity under tensor products (Proposition~\ref{prop: M add}): $\cM(\cN_1\ox\cN_2)=\cM(\cN_1)+\cM(\cN_2)$.
\item Subadditivity under serial composition of channels (Proposition~\ref{prop: M subadd}): $\cM(\cN_2\circ\cN_1)\leq \cM(\cN_1)+\cM(\cN_2)$.
\item  Faithfulness (Proposition~\ref{prop:faithfulness-mana}): $\cM(\cN) \geq 0$ and $\cM(\cN)=0$ if and only if $\cN \in \rm{CPWP}$.
\item Amortization inequality (Proposition~\ref{prop: M amo}): $\forall \rho_{RA}$, \ $\cM((\operatorname{id}_R\ox \cN)(\rho_{RA}))-\cM(\rho_{RA})\le \cM(\cN)$.
\item Monotonicity under CPWP superchannels (Proposition~\ref{prop: M monotone}), which implies monotonicity under completely stabilizer-preserving superchannels.
\end{enumerate}  

%%%%%%%%%%%%%%%%%%%%%%%%%%%%%%%%%%%%

\begin{proposition}[Reduction to states]
\label{prop:mana-reduction-to-states}
Let $\cN$ be a replacer channel,  acting as $\cN(\rho) = \tr[\rho]\sigma$ for an arbitrary input state $\rho$, with $\sigma$ a state. Then
\begin{equation}
\cM(\cN) = \cM(\sigma).
\end{equation}
\end{proposition}

\begin{proof}
Applying definitions and the fact that $\tr[A_\bu]=1$ for a phase-space point operator $A_{\bu}$, we find that
\begin{equation}
\cM(\cN) = \log \max_{\bu} \| \cN(A_{\bu})\|_{W,1} = \log \max_{\bu} \| \tr[A_{\bu}] \sigma \|_{W,1} = \log  \|  \sigma \|_{W,1} = \cM(\sigma),
\end{equation}
concluding the proof.
\end{proof}

\begin{proposition}[Additivity]
\label{prop: M add}
For quantum channels $\cN_1$ and $\cN_2$, the following additivity identity holds
\begin{align}
\cM(\cN_1\ox\cN_2)=\cM(\cN_1)+\cM(\cN_2).
\end{align}
More generally, the same additivity identity holds if $\cN_1$ and $\cN_2$ are Hermiticity-preserving linear maps.
\end{proposition}
\begin{proof}
The proof relies on basic properties of the Wigner 1-norm and composite phase-space point operators, i.e.,
\begin{align}
\cM(\cN_1\ox \cN_2) 
&= \log \max_{\bu_1,\bu_2} \|(\cN_1\ox\cN_2)(A_{\bu_1}\ox A_{\bu_2})\|_{W,1}\\
& = \log \max_{\bu_1,\bu_2} \left[\|\cN_1(A_{\bu_1})\|_{W,1} \cdot \|\cN_2(A_{\bu_2})\|_{W,1}\right]\\
& = \log \max_{\bu_1} \|\cN_1(A_{\bu_1})\|_{W,1} + \log \max_{\bu_2} \|\cN_2(A_{\bu_2})\|_{W,1}\\
& = \cM(\cN_1)+\cM(\cN_2).
\end{align}
This concludes the proof.
\end{proof}

\begin{proposition}[Subadditivity]
\label{prop: M subadd}
For quantum channels $\cN_1$ and $\cN_2$, the following subadditivity inequality holds
\begin{align}
\cM(\cN_2\circ \cN_1)\leq \cM(\cN_1)+\cM(\cN_2).
\end{align}
More generally, the same subadditivity inequality holds if $\cN_1$ and $\cN_2$ are Hermiticity-preserving linear maps.
\end{proposition}
\begin{proof}
Consider the following for an arbitrary phase-space point operator $A_\bu$:
\begin{align}
\log \| (\cN_2\circ \cN_1)(A_\bu)\|_{W,1} & = \log \frac{\| (\cN_2\circ \cN_1)(A_\bu)\|_{W,1}}{\|\cN_1(A_\bu)\|_{W,1}}
+ \log \|\cN_1(A_\bu)\|_{W,1} \\
& = \log \frac{\left\| \cN_2\left( \sum_{\bu'}W_{\cN_1(A_\bu)}(\bu') A_{\bu'} \right)\right\|_{W,1}}{\sum_{\bu'}|W_{\cN_1(A_\bu)}(\bu')|}
+ \log \|\cN_1(A_\bu)\|_{W,1} \\
& \leq \log \sum_{\bu'} \frac{|W_{\cN_1(A_\bu)}(\bu')|}{\sum_{\bu'}|W_{\cN_1(A_\bu)}(\bu')|} \left\| \cN_2\left(  A_{\bu'} \right)\right\|_{W,1}
+ \log \|\cN_1(A_\bu)\|_{W,1}\\
& \leq \log \max_{\bu'} \left\| \cN_2\left(  A_{\bu'} \right)\right\|_{W,1}
+ \log \max_{\bu}\|\cN_1(A_\bu)\|_{W,1} \\
& = \cM(\cN_1)+\cM(\cN_2).
\end{align}
Since the chain of inequalities holds for an arbitrary phase-space point operator $A_\bu$, we conclude the statement of the proposition.
\end{proof}

\begin{proposition}[Faithfulness]
\label{prop:faithfulness-mana}
Let $\cN_{A \to B}$ be a quantum channel. Then the mana of the channel $\cN$ satisfies $\cM(\cN) \geq 0$, and $\cM(\cN)=0$ if and only if $\cN \in \rm{CPWP}$.
\end{proposition}

\begin{proof}
To see the first claim, from the assumption that $\cN$ is a quantum channel and \eqref{eq:normalized-dWF-channel}, we find that
\begin{align}
\sum_{\bv} |\cW_{\cN}(\bv|\bu)| = 2 \left[ \sum_{\bv : \cW_{\cN}(\bv|\bu) < 0}|\cW_{\cN}(\bv|\bu)|\right] + 1 \geq 1 \quad \forall \bu.
\label{eq:key-for-faithful}
\end{align}
Taking a maximization over $\bu$ and applying a logarithm leads to the conclusion that $\cM(\cN) \geq 0$ for all channels $\cN$.

Now suppose that $\cN \in \rm{CPWP}$. Then by Theorem~\ref{th: PWP}, it follows that $\cW_{\cN}(\bv|\bu)$ is a conditional probability distribution, so that $\sum_{\bv} |\cW_{\cN}(\bv|\bu)| = \sum_{\bv} \cW_{\cN}(\bv|\bu) = 1$ for all $\bu$. It then follows from the definition that $\cM(\cN)=0$.

Finally, suppose that $\cM(\cN)=0$. By definition, this implies that $\max_{\bu } \sum_{\bv} |\cW_{\cN}(\bv|\bu)| = 1$. However, consider that the rightmost inequality in \eqref{eq:key-for-faithful} holds for all channels. So our assumption and this inequality imply that 
$\sum_{\bv : \cW_{\cN}(\bv|\bu) < 0}|\cW_{\cN}(\bv|\bu)| = 0$ for all $\bu$, which means that $\cW_{\cN}(\bv|\bu) \geq 0$ for all $\bu, \bv$. By Theorem~\ref{th: PWP}, it follows that $\cN \in \rm{CPWP}$.
\end{proof}

\begin{proposition}[Amortization inequality]\label{prop: M amo}
For any quantum channel $\cN_{A\to B}$, the following inequality holds
\begin{align}
\sup_{\rho_A} [\cM(\cN(\rho_A))-\cM(\rho_A)]\le \cM(\cN).
\label{eq:amortize-no-ref}
\end{align}
Furthermore, we have that
\begin{align}
\sup_{\rho_{RA}}[\cM((\operatorname{id}_R\ox \cN)(\rho_{RA}))-\cM(\rho_{RA})]\le \cM(\cN).
\label{eq:amortize-with-ref}
\end{align}
\end{proposition}
\begin{proof}
The inequality in \eqref{eq:amortize-no-ref} is a direct consequence of reduction to states (Proposition~\ref{prop:mana-reduction-to-states}) and subadditivity of mana with respect to serial compositions (Proposition~\ref{prop: M subadd}). Indeed, letting $\cN'$ be a replacer channel that prepares the state $\rho_A$, we find that
\begin{equation}
\cM(\cN(\rho_A)) = \cM(\cN\circ \cN') \leq \cM(\cN) + \cM(\cN') = \cM(\cN) + \cM(\rho_A). \label{eq:amortize-no-ref-no-opt}
\end{equation}
for all input states $\rho_A$,
from which we conclude \eqref{eq:amortize-no-ref}.

  By applying the inequality in \eqref{eq:amortize-no-ref-no-opt} with the substitution $\cN \to \operatorname{id}\otimes \cN$, the additivity of the mana of a channel from Proposition~\ref{prop: M add}, and the fact that the identity channel is free (and thus has mana equal to zero), we finally conclude that
  \begin{align}
 \cM((\operatorname{id}_R\ox \cN)(\rho_{RA}))-\cM(\rho_{RA}) &\le \cM(\operatorname{id}_R\ox \cN)\\
 &=\cM(\operatorname{id}_R)+\cM(\cN)\\
 &= \cM(\cN),
  \end{align}
  from which we conclude \eqref{eq:amortize-with-ref}.
  \end{proof}

\begin{theorem}[Monotonicity]\label{prop: M monotone}
Let $\mathcal{N}_{A\rightarrow B}$ be a quantum channel, and let
$\Xi^{\operatorname{CPWP}}$ be a CPWP\ superchannel as given in Definition~\ref{def:CPWP-superch}. Then
$\mathcal{M}(\mathcal{N})$ is a channel magic measure in the sense that%
\begin{equation}
\mathcal{M}(\mathcal{N})\geq\mathcal{M}(\Xi^{\operatorname{CPWP}}(\mathcal{N})).
\label{eq:mono-mana-CPWP}
\end{equation}

\end{theorem}

\begin{proof}
Recalling the definition of the channel mana in terms of the discrete Wigner
function (see~\eqref{eq:dWf-mana-def}) and abbreviating $\Xi^{\text{CPWP}}$ as $\Xi$, consider that%
\begin{align}
\mathcal{M}(\Xi^{\text{CPWP}}(\mathcal{N}))  &  =\log\max_{\boldsymbol{u}%
_{C}}\sum_{\boldsymbol{v}_{D}}\left\vert W_{\Xi(\mathcal{N})}%
(\boldsymbol{v}_{D}|\boldsymbol{u}_{C})\right\vert \\
&  =\log\max_{\boldsymbol{u}_{C}}\sum_{\boldsymbol{v}_{D}}\left\vert
\sum_{\boldsymbol{u}_{A},\boldsymbol{v}_{B}}W_{\Xi}(\boldsymbol{u}%
_{A},\boldsymbol{v}_{D}|\boldsymbol{u}_{C},\boldsymbol{v}_{B})W_{\mathcal{N}%
}(\boldsymbol{v}_{B}|\boldsymbol{u}_{A})\right\vert \\
&  \leq\log\max_{\boldsymbol{u}_{C}}\sum_{\boldsymbol{v}_{D},\boldsymbol{u}%
_{A},\boldsymbol{v}_{B}}\left\vert W_{\Xi}(\boldsymbol{u}_{A}%
,\boldsymbol{v}_{D}|\boldsymbol{u}_{C},\boldsymbol{v}_{B})W_{\mathcal{N}%
}(\boldsymbol{v}_{B}|\boldsymbol{u}_{A})\right\vert \\
&  =\log\max_{\boldsymbol{u}_{C}}\sum_{\boldsymbol{v}_{D},\boldsymbol{u}%
_{A},\boldsymbol{v}_{B}}W_{\Xi}(\boldsymbol{u}_{A},\boldsymbol{v}%
_{D}|\boldsymbol{u}_{C},\boldsymbol{v}_{B})\left\vert W_{\mathcal{N}%
}(\boldsymbol{v}_{B}|\boldsymbol{u}_{A})\right\vert \\
&  =\log\max_{\boldsymbol{u}_{C}}\sum_{\boldsymbol{u}_{A},\boldsymbol{v}_{B}%
}W_{\Xi}(\boldsymbol{u}_{A}|\boldsymbol{u}_{C},\boldsymbol{v}%
_{B})\left\vert W_{\mathcal{N}}(\boldsymbol{v}_{B}|\boldsymbol{u}%
_{A})\right\vert \\
&  =\log\max_{\boldsymbol{u}_{C}}\sum_{\boldsymbol{u}_{A},\boldsymbol{v}_{B}%
}W_{\Xi}(\boldsymbol{u}_{A}|\boldsymbol{u}_{C})\left\vert W_{\mathcal{N}%
}(\boldsymbol{v}_{B}|\boldsymbol{u}_{A})\right\vert .
\label{eq:ch-mana-monotone-CPWP-superch}%
\end{align}
The second equality follows from \eqref{eq:output-dWf-superch-dWf}. The first
inequality follows from the triangle inequality. The third equality follows
from the assumption that the superchannel $\Xi$ is CPWP, so that its
discrete Wigner function is non-negative (see Theorem~\ref{thm:CPWP-superchs-equiv}). The fourth equality follows from
marginalizing $W_{\Xi}$ over $\boldsymbol{v}_{D}$. The fifth equality
follows from the non-signaling constraint in \eqref{eq:non-sig-dWf-superch}.
Continuing, we find that%
\begin{align}
\text{Eq.~\eqref{eq:ch-mana-monotone-CPWP-superch}}  &  =\log\max
_{\boldsymbol{u}_{C}}\sum_{\boldsymbol{u}_{A}}W_{\Xi}(\boldsymbol{u}%
_{A}|\boldsymbol{u}_{C})\sum_{\boldsymbol{v}_{B}}\left\vert W_{\mathcal{N}%
}(\boldsymbol{v}_{B}|\boldsymbol{u}_{A})\right\vert \\
&  \leq\log\max_{\boldsymbol{u}_{C}}\sum_{\boldsymbol{u}_{A}}W_{\Xi
}(\boldsymbol{u}_{A}|\boldsymbol{u}_{C})\left[  \max_{\boldsymbol{u}_{A}}%
\sum_{\boldsymbol{v}_{B}}\left\vert W_{\mathcal{N}}(\boldsymbol{v}%
_{B}|\boldsymbol{u}_{A})\right\vert \right] \\
&  =\log\max_{\boldsymbol{u}_{A}}\sum_{\boldsymbol{v}_{B}}\left\vert
W_{\mathcal{N}}(\boldsymbol{v}_{B}|\boldsymbol{u}_{A})\right\vert \\
&  =\mathcal{M}(\mathcal{N}).
\end{align}
The first equality follows from rearranging sums. The  inequality follows
from bounding $\sum_{\boldsymbol{v}_{B}}\left\vert W_{\mathcal{N}%
}(\boldsymbol{v}_{B}|\boldsymbol{u}_{A})\right\vert $ in terms of its maximum
value (so that it is no longer dependent on $\boldsymbol{u}_{A}$). The
penultimate equality follows because $\sum_{\boldsymbol{u}_{A}}W_{\Xi
}(\boldsymbol{u}_{A}|\boldsymbol{u}_{C})=1$, and the final one follows by definition.
\end{proof}

\begin{remark}
\label{rem:mono-extend-to-CP}
We note here that the monotonicity inequality in \eqref{eq:mono-mana-CPWP} holds more generally if $\mathcal{N}$ is a completely positive map that is not necessarily trace preserving.
\end{remark}

\subsection{Generalized thauma of a quantum channel}

\label{sec:gen-thauma}

In this section, we define a rather general measure of magic for a quantum channel, called the generalized thauma, which extends to channels the definition from \cite{Wang2018} for states. To define it, recall that a generalized divergence  $\mathbf{D}(\rho\|\sigma)$ is any function of a quantum state $\rho$ and a positive semi-definite operator $\sigma$ that obeys data processing \cite{Polyanskiy2010b,SW12}, i.e., $\mathbf{D}(\rho\|\sigma) \geq \mathbf{D}(\cN(\rho)\|\cN(\sigma))$ where $\cN$ is a quantum channel. 
Examples of generalized divergences, in addition to the trace distance and relative entropy, include the Petz--R\'enyi relative entropies \cite{P86}, the sandwiched R\'enyi relative entropies \cite{Muller-Lennert2013,Wilde2014a}, the Hilbert $\alpha$-divergences \cite{BG17}, and the $\chi^2$ divergences \cite{TKRWV10}.
One can then define the generalized channel divergence \cite{LKDW18}, as a way of quantifying the distinguishability of two quantum channels $\cN_{A \to B}$ and $\cP_{A \to B}$, as follows:
\begin{equation}
\mathbf{D}(\cN \| \cP) \coloneqq \sup_{\psi_{RA}} \mathbf{D}(\cN_{A \to B}(\psi_{RA})\|\cP_{A \to B}(\psi_{RA})),
\label{eq:gen-div}
\end{equation}
where the optimization is with respect to all pure states $\psi_{RA}$ such that system $R$ is isomorphic to the channel input system $A$ (note that one does not achieve a higher value of $\mathbf{D}(\cN \| \cP)$ by allowing for an optimization over mixed states $\rho_{RA}$ with an arbitrarily large reference system \cite{LKDW18}, as a consequence of purification, the Schmidt decomposition theorem, and data processing). More generally, $\cP_{A \to B}$ can be a completely positive map in the definition in \eqref{eq:gen-div}. Interestingly, the generalized channel divergence is monotone under the action of a superchannel $\Xi$:
\begin{equation}
\mathbf{D}(\cN \| \cP) \geq \mathbf{D}(\Xi(\cN) \| \Xi(\cP)),
\end{equation}
as shown in \cite[Section~V-A]{Gour2018}.

We then define generalized thauma as follows:
\begin{definition}[Generalized thauma of a quantum channel]
The generalized thauma of a quantum channel $\cN_{A \to B}$ is defined as
\begin{equation}
\boldsymbol{\theta}(\cN) \coloneqq \inf_{\cE : \cM(\cE)\le 0} \mathbf{D}(\cN \| \cE),
\label{eq:gen-thauma-channels}
\end{equation}
where the optimization is with respect to all completely positive maps $\cE$ having mana $\cM(\cE)\le 0$.
\end{definition}

It is clear that the above definition extends the generalized thauma of a state \cite{Wang2018}, which we recall is given by
\begin{equation}
\boldsymbol{\theta}(\rho) \coloneqq \inf_{\sigma\geq 0 : \cM(\sigma)\le 0} \mathbf{D}(\rho \| \sigma).
\label{eq:GT-state}
\end{equation}
We now prove that the generalized thauma of a quantum channel reduces to the state measure whenever the channel $\cN$ is a replacer channel:

\begin{proposition}[Reduction to states]
\label{prop:gen-thauma-reduction-to-states}
Let $\cN$ be a replacer channel,  acting as $\cN(\rho) = \tr[\rho]\sigma$ for an arbitrary input state $\rho$, where $\sigma$ is a state. Then
\begin{equation}
\boldsymbol{\theta}(\cN) = \boldsymbol{\theta}(\sigma).
\end{equation}
\end{proposition}

\begin{proof}
First, denoting the maximally mixed state by $\pi$, consider that
\begin{align}
\boldsymbol{\theta}(\cN) & = \inf_{\cE : \cM(\cE) \leq 0} \sup_{\psi_{RA}}\mathbf{D}((\operatorname{id}_R\otimes \cN)(\psi_{RA}) \| (\operatorname{id}_R\otimes \cE)(\psi_{RA})) \\
& \geq  \inf_{\cE : \cM(\cE) \leq 0} \mathbf{D}(\pi_{R}\ox \cN(\pi_A) \| \pi_{R}\ox \cE(\pi_{A})) \\
& = \inf_{\cE : \cM(\cE) \leq 0} \mathbf{D}( \sigma_B \|  \cE(\pi)) \\
& = \inf_{\omega : \cM(\omega) \leq 0} \mathbf{D}( \sigma_B \|  \omega) \\
& = \boldsymbol{\theta}(\sigma).
\end{align}
The first equality follows from the definition. The inequality follows by choosing the input state suboptimally to be $\pi_R \ox \pi_A$. The second equality follows because the generalized divergence is invariant with respect to tensoring in the same state for both arguments. The third equality follows because $\pi$ is a free state with non-negative Wigner function and $\cE$ is a completely positive map with $\cM(\cE) \leq 0$. Since one can reach all and only the operators $\omega \in \cW$, the equality follows. Then the last equality follows from the definition.

To see the other inequality, consider that $\cE(\rho) = \tr[\rho] \omega$, for $\omega \in \cW$, is a particular completely positive map satisfying $\cM(\cE) = \cM(\omega) \leq 0$, so that
\begin{align}
\boldsymbol{\theta}(\cN) & = \inf_{\cE : \cM(\cE) \leq 0} \sup_{\psi_{RA}}\mathbf{D}((\operatorname{id}_R\otimes \cN)(\psi_{RA}) \| (\operatorname{id}_R\otimes \cE)(\psi_{RA})) \\
& \leq \inf_{\omega : \cM(\omega) \leq 0} \mathbf{D}(\psi_R \otimes \sigma_B \| \psi_R\otimes \omega_B) \\
& = \inf_{\omega : \cM(\omega) \leq 0} \mathbf{D}( \sigma_B \|  \omega_B) \\
& = \boldsymbol{\theta}(\sigma).
\end{align}
This concludes the proof.
\end{proof}

\bigskip
That the generalized thauma of channels proposed in \eqref{eq:gen-thauma-channels} is a good measure of magic for quantum channels is a consequence of the following proposition:
 \begin{theorem}[Monotonicity]
 \label{thm:gen-thauma-monotone-PWP-superch}
 Let $\cN_{A\to B}$ be a quantum channel, and let $\Xi^{\operatorname{CPWP}}$ be a CPWP superchannel as given in Definition~\ref{def:CPWP-superch}.
Then $\boldsymbol{\theta}(\cN)$ is a channel magic measure in the sense that
\begin{align}
\boldsymbol{\theta}(\cN_{A \to B}) \geq \boldsymbol{\theta}(\Xi^{\operatorname{CPWP}}(\cN_{A\to B}))\ .
\end{align}
\end{theorem}
\begin{proof}
The idea is to utilize the generalized divergence and its basic property of data processing. In more detail, consider that
\begin{align}
\boldsymbol{\theta}(\cN_{A\to B})&= \inf_{\cE : \cM(\cE)\le 0} \mathbf{D}(\cN_{A\to B}\|\cE)\\
&\ge \inf_{\cE : \cM(\cE)\le 0} \mathbf{D}(\Xi^{\operatorname{CPWP}}(\cN_{A\to B})\|\Xi^{\operatorname{CPWP}}(\cE))\\
&\ge \inf_{\widehat\cE : \cM(\widehat\cE)\le 0} \mathbf{D}(\Xi^{\operatorname{CPWP}}(\cN_{A\to B})\|\widehat\cE)\\
&= \boldsymbol{\theta}(\Xi^{\operatorname{CPWP}}(\cN_{A\to B})).
\end{align}
The first inequality follows from the fact that the generalized divergence of channels is monotone under the action of a superchannel \cite[Section~V-A]{Gour2018}. The second inequality follows from the monotonicity of $\cM(\cN)$ given in Theorem~\ref{prop: M monotone} (and which extends more generally to completely positive maps as stated in Remark~\ref{rem:mono-extend-to-CP}). This monotonicity implies that $\cM(\cE) \geq \cM(\Xi^{\operatorname{CPWP}}(\cE))$ and leads to the second inequality.
\end{proof}

\bigskip

A generalized divergence is called \textit{strongly faithful} \cite{BHKW18} if for a state $\rho_A$ and a subnormalized state $\sigma_A$, we have
$\mathbf{D}(\rho_A\|\sigma_A)\geq 0$ in general and
$\mathbf{D}(\rho_A\|\sigma_A)=0$ if and only if $\rho_A=\sigma_A$. 

\begin{proposition}[Faithfulness]
\label{eq:faithful-gen-thauma}
Let $\mathbf{D}$ be a strongly faithful generalized divergence. Then the generalized thauma $\boldsymbol{\theta}(\cN)$ of a channel $\cN$ defined through $\mathbf{D}$ is non-negative and it is equal to zero if $\cN \in \rm{CPWP}$. If the generalized divergence is furthermore continuous and $\boldsymbol{\theta}(\cN) = 0$, then $\cN \in \rm{CPWP}$.
\end{proposition}

\begin{proof}
From Lemma~\ref{lem:mana-CP-maps} in Appendix~\ref{app:tni-lemma}, it follows that any completely positive map $\cE$ subject to the constraint $\cM(\cE) \leq 0$ is trace non-increasing on the set $\cW_+$. It thus follows that $\cE_{A \to B}(\psi_{RA})$ is subnormalized for any input state $\psi_{RA} \in \cW_+$. By restricting the maximization to such input states in $\cW_+$, applying the faithfulness assumption, and applying the definition of generalized thauma, we conclude that $\boldsymbol{\theta}(\cN) \geq 0$.

Suppose that $\cN \in \rm{CPWP}$. Then by Proposition~\ref{prop:faithfulness-mana}, $\cM(\cN) = 0$ and so we can set $\cE = \cN$ in the definition of generalized thauma and conclude from the faithfulness assumption that $\boldsymbol{\theta}(\cN) = 0$.

Finally, suppose that $\boldsymbol{\theta}(\cN) = 0$. By the assumption of continuity, this means that there exists a completely positive map $\cE$ satisfying  $\mathbf{D}(\cN \| \cE) = 0$. By Lemma~\ref{lem:mana-CP-maps} in Appendix~\ref{app:tni-lemma} and the faithfulness assumption, this in turn means that $\cN_{A\to B}(\Phi_{RA}) = \cE_{A\to B}(\Phi_{RA})$ for the maximally entangled state $\Phi_{RA}\in\cW_+$, which implies that $\cN_{A\to B} = \cE_{A\to B}$. However, we have that $\cM(\cE) \leq 0$, implying that $\cM(\cN) = 0$, since $\cN$ is a channel and $\cM(\cN) \geq 0$ for all channels. By Proposition~\ref{prop:faithfulness-mana}, we conclude that $\cN \in \rm{CPWP}$.
\end{proof}

\bigskip
As discussed in \cite{LKDW18,SWAT17}, a generalized  divergence possesses the direct-sum property  on classical-quantum states if the following equality holds:
\begin{equation}
\mathbf{D}\!\left(\sum_{x}p_{X}(x)|x\rangle\langle x|_{X}\otimes\rho^{x}\middle\Vert\sum_{x}p_{X}(x)|x\rangle\langle x|_{X}\otimes\sigma^{x}\right)  =\sum_{x}p_{X}(x)\mathbf{D}(\rho^{x}\Vert\sigma^{x}),\label{eq:d-sum-prop}
\end{equation}
where $p_{X}$ is a probability distribution, $\{|x\rangle\}_{x}$ is an orthonormal basis, and $\{\rho^{x}\}_{x}$ and $\{\sigma^{x}\}_{x}$ are sets of states. We note that this property holds for trace distance, quantum relative entropy \cite{Umegaki1962}, and the Petz-R\'enyi \cite{P86} and sandwiched R\'enyi \cite{Muller-Lennert2013,Wilde2014a} quasi-entropies $\operatorname{sgn}(\alpha-1)\tr\left[\rho^\alpha\sigma^{1-\alpha}\right]$ and $\operatorname{sgn}(\alpha-1)\tr[(\sigma^{\frac{1-\alpha}{2\alpha}}\rho\sigma^{\frac{1-\alpha}{2\alpha}})^\alpha]$, respectively.

For such generalized divergences, which are additionally continuous, as well as convex in the second argument, we find that an exchange of the minimization and the maximization in the definition of the generalized thauma is possible:
\begin{proposition}[Minimax]
\label{prop:minimax-gen-thauma}
Let $\mathbf{D}$ be a generalized divergence that is continuous, obeys the direct-sum property in \eqref{eq:d-sum-prop}, and is convex in the second argument. Then the following exchange of min and max is possible in the generalized thauma:
\begin{equation}
\boldsymbol{\theta}(\cN) = \sup_{\psi_{RA}} \inf_{\cE : \cM(\cE) \leq 0 } \mathbf{D}(\cN_{A \to B}(\psi_{RA}) \| \cE_{A \to B}(\psi_{RA})).
\end{equation}
\end{proposition}
\begin{proof}
Let $\mathcal{E}$ be a fixed completely positive map such that $\mathcal{M}%
(\mathcal{E})\leq0$. Let $\psi_{RA}^{1}$ and $\psi_{RA}^{2}$ be input states to consider for the maximization.
Due to the unitary freedom of purifications and invariance of generalized divergence with respect to unitaries, we can equivalently consider the maximization to be over the convex set of density operators acting on the channel input system~$A$.
Define
\begin{equation}
\rho_{A}^{\lambda}=\lambda\psi_{A}^{1}+(1-\lambda)\psi_{A}^{2},
\end{equation}
for $\lambda\in\lbrack0,1]$. Then the state%
\begin{equation}
|\phi^{\lambda}\rangle_{R^{\prime}RA}:=\sqrt{\lambda}|0\rangle_{R^{\prime}%
}|\psi^{1}\rangle_{RA}+\sqrt{1-\lambda}|1\rangle_{R^{\prime}}|\psi^{2}%
\rangle_{RA}%
\end{equation}
purifies $\rho_{A}^{\lambda}$ and is related to a purification $|\psi^{\lambda
}\rangle_{RA}$ of $\rho_{A}^{\lambda}$ by an isometry. It then follows that%
\begin{align}
& \mathbf{D}(\mathcal{N}_{A\rightarrow B}(\psi_{RA}^{\lambda})\Vert
\mathcal{E}_{A\rightarrow B}(\psi_{RA}^{\lambda}))\nonumber\\
& =\mathbf{D}(\mathcal{N}_{A\rightarrow B}(\phi_{R^{\prime}RA}^{\lambda}%
)\Vert\mathcal{E}_{A\rightarrow B}(\phi_{R^{\prime}RA}^{\lambda}))\\
& \geq\lambda\mathbf{D}(|0\rangle\langle0|_{R^{\prime}}\otimes\mathcal{N}%
_{A\rightarrow B}(\psi_{RA}^{1})\Vert|0\rangle\langle0|_{R^{\prime}}%
\otimes\mathcal{E}_{A\rightarrow B}(\psi_{RA}^{1}))\nonumber\\
& \qquad +(1-\lambda)\mathbf{D}(|1\rangle\langle1|_{R^{\prime}}\otimes\mathcal{N}%
_{A\rightarrow B}(\psi_{RA}^{2})\Vert|1\rangle\langle1|_{R^{\prime}}%
\otimes\mathcal{E}_{A\rightarrow B}(\psi_{RA}^{2}))\\
& =\lambda\mathbf{D}(\mathcal{N}_{A\rightarrow B}(\psi_{RA}^{1})\Vert
\mathcal{E}_{A\rightarrow B}(\psi_{RA}^{1}))+(1-\lambda)\mathbf{D}%
(\mathcal{N}_{A\rightarrow B}(\psi_{RA}^{2})\Vert\mathcal{E}_{A\rightarrow
B}(\psi_{RA}^{2})).
\end{align}
The inequality follows from data processing, by applying a completely dephasing channel to the register $R'$. The last equality again follows from data processing. So the objective function is concave in the argument being maximized (again thinking of the maximization being performed over density operators on $A$ rather than pure states on $RA$).

By assumption, for a fixed input state $\psi_{RA}$, the objective function is convex in the second argument and the set of completely positive maps $\cE$ satisfying $\cM(\cE)$ is convex.

Then the Sion minimax theorem \cite{sion58} applies, and we conclude the statement of the proposition.
\end{proof}

\begin{remark}
Examples of generalized divergences to which Proposition~\ref{prop:minimax-gen-thauma} applies include the quantum relative entropy \cite{Umegaki1962}, the sandwiched R\'enyi relative entropy \cite{Muller-Lennert2013,Wilde2014a}, and the Petz--R\'enyi relative entropy \cite{P86}. The proposition applies to the latter two by working with the corresponding quasi-entropies and then lifting the result to the actual relative entropies.
\end{remark}

\subsection{Max-thauma of a quantum channel}

As a particular case of the generalized thauma of a quantum channel defined in \eqref{eq:gen-thauma-channels}, we consider the max-thauma of a quantum channel, which is the max-relative entropy divergence between the channel and the set of completely positive maps with non-positive mana.
Specifically, for a given quantum channel $\cN_{A\to B}$, the max-thauma of $\cN_{A\to B}$ is defined by
\begin{align}\label{eq:theta divergence}
\theta_{\max}(\cN)\coloneqq \min_{\cE : \cM(\cE)\le 0} D_{\max}(\cN\|\cE),
\end{align}
where the minimum is taken with respect to all completely positive maps $\cE$ satisfying $\cM(\cE)\le 0$ and
\begin{equation}
D_{\max}(\cN \| \cE) \coloneqq \sup_{\psi_{RA}} D_{\max}(\cN_{A\to B}(\psi_{RA})\| \cE_{A\to B}(\psi_{RA}))
\label{eq:max-div-channels-1}
\end{equation}
is the max-divergence of channels \cite{Cooney2016}. (More generally, $\cN$ and $\cE$ could be arbitrary completely positive maps  in \eqref{eq:max-div-channels-1}.) Note that it is known that \cite{Diaz2018,BHKW18}
\begin{equation}
D_{\max}(\cN\|\cE) = D_{\max}(\cN_{A\to B}(\Phi_{RA})\| \cE_{A\to B}(\Phi_{RA})) = \log \min \{t: J_{AB}^\cN \le t J_{AB}^{\cE}\},
\label{eq:max-rel-ent-channels}
\end{equation}
where $\Phi_{RA}$ is the maximally entangled state and $J_{AB}^\cN$ is the \Choi matrix of the channel $\cN_{A\to B}$ and similarly for $J_{AB}^{\cE}$.

Due to the properties of max-relative entropy, it follows that Theorem~\ref{thm:gen-thauma-monotone-PWP-superch} and Propositions~\ref{prop:gen-thauma-reduction-to-states}, \ref{eq:faithful-gen-thauma}, and~\ref{prop:minimax-gen-thauma} apply to the max-thauma of a channel, implying reduction to states, that it is monotone with respect to completely CPWP superchannels, faithful, and obeys a minimax theorem, so that
\begin{align}
\theta_{\max}(\cR) & = \theta_{\max}(\sigma), \ \text{ if } \ \cR(\rho) = \tr[\rho] \sigma \label{prop:max-thauma-reduction-to-states}\\
\theta_{\max}(\cN) & \geq \theta_{\max}(\Xi^{\operatorname{CPWP}}(\cN)), \\
\theta_{\max}(\cN) & \geq 0 \quad \text{  and  } \quad \theta_{\max}(\cN) = 0 \quad \text{  if and only if  }\quad \cN \in \rm{CPWP} ,
\label{eq:max-thauma-faithful}
\\
\theta_{\max}(\cN) & = \min_{\cE : \cM(\cE)\le 0}\max_{\psi_{RA}} D_{\max}(\cN_{A\to B}(\psi_{RA})\| \cE_{A\to B}(\psi_{RA})) \\
& = \max_{\psi_{RA}} \min_{\cE : \cM(\cE)\le 0} D_{\max}(\cN_{A\to B}(\psi_{RA})\| \cE_{A\to B}(\psi_{RA})),
\end{align}
where $\cN$ is a quantum channel and $\Xi^{\operatorname{CPWP}}$ is a CPWP superchannel.
 
 We can alternatively express the max-thauma of a channel as the following SDP:
  \begin{proposition}[SDP for max-thauma] 
  \label{prop:SDPs-max-thauma}
For a given quantum channel $\cN_{A\to B}$, its max-thauma $\theta_{\max}(\cN)$ can be written as the following SDP:
 \begin{equation}\label{eq:theta max channel min}
\begin{split}
\theta_{\max}(\cN) =\log \min \ & t \\
\text{s.t. } & J_{AB}^\cN \le Y_{AB}\\
& \sum_{\bv}|\tr [(A_A^{\bu}\ox A_B^{\bv})Y_{AB}]|/d_B \le t, \quad\forall \bu,
\end{split}
\end{equation}
where $J_{AB}^\cN$ is the \Choi matrix of the channel $\cN_{A\to B}$. Moreover, the dual SDP to the above is as follows:
\begin{equation}
\begin{split}
\theta_{\max}(\cN) =\log \max \ & \tr [J_{AB}^\cN V_{AB}] \\
\text{s.t. } &  \sum_\bv b_{\bv}\le 1\\
&0\le V_{AB}\le \sum_{\bv,\bu} (c_{\bv,\bu}-f_{\bv,\bu})A_A^{\bv}\ox A_B^{\bu}/d_{B},\\
&  c_{\bv,\bu} +f_{\bv,\bu} \le b_{\bv},\quad \forall \bu, \bv,\\
&  c_{\bv,\bu}\ge 0, f_{\bv,\bu}\ge 0, \quad \forall \bu, \bv.
\end{split}
\end{equation}
%\begin{equation}\label{eq:theta max channel max}
%\begin{split}
%\theta_{\max}(\cN) =\log \max \ & \tr [J_{AB}^\cN R_{AB}] \\
%\text{s.t. } &  \sum_\bv b_{\bv}\le 1,\\
%& R_{AB}\ge 0, \ |W_{R_{AB}}(\bv,\bu)| \le b_{\bv}/d_B, \quad\forall \bu, \bv.
%\end{split}
%\end{equation}
 \end{proposition}
 
\begin{proof}
Consider the following chain of equalities:
\begin{align}
\theta_{\max}(\cN) &=\log  \min_{\cE:\cM(\cE)\le 0} D_{\max}(\cN\|\cE)\\
&=\log \min \left\{t: J_{AB}^\cN  \le tJ_{AB}^{\cE'},\  \cM(\cE')\le 0 \right\}\\
&=\log \min \left\{t: J_{AB}^\cN \le tJ_{AB}^{\cE'}, \sum_{\bv}|\tr [\cE'(A_A^\bu)A_B^\bv]|/d_B\le 1, \forall \bu \right\}\\
&=\log \min \left\{t: J_{AB}^\cN \le J_{AB}^\cE, \sum_{\bv}|\tr [\cE(A_A^\bu)A_B^\bv]|/d_B\le t, \forall \bu \right\}\\
&=\log \min \left\{t: J_{AB}^\cN \le Y_{AB}, \sum_{\bv}|\tr [(A_A^{\bu}\ox A_B^{\bv})Y_{AB}]|/d_B\le t, \forall \bu \right\}, \label{eq:abs theta max}
\end{align}
where the second equality follows from \eqref{eq:max-rel-ent-channels} and the last from the fact that $\cE$ is completely positive and thus in one-to-one correspondence with positive semi-definite bipartite operators.
We further rewrite the absolute-value constraint in \eqref{eq:abs theta max} and arrive at the following SDP: 
\begin{align}
\theta_{\max}(\cN)& =\log \min \left\{t: J_{AB}^\cN \le Y_{AB}, -t\le \sum_{\bv} \tr [(A_A^{\bu}\ox A_B^{\bv})Y_{AB}]/d_B\le t, \forall \bu \right\},
\end{align}

Then we use the Lagrangian method to obtain the dual SDP:
\begin{equation}
\begin{split}
\theta_{\max}(\cN) =\log \max \ & \tr [J_{AB}^\cN V_{AB}] \\
\text{s.t. } &  \sum_\bv b_{\bv}\le 1\\
&0\le V_{AB}\le \sum_{\bv,\bu} (c_{\bv,\bu}-f_{\bv,\bu})A_A^{\bv}\ox A_B^{\bu}/d_{B},\\
&  c_{\bv,\bu} +f_{\bv,\bu} \le b_{\bv},\quad \forall \bu, \bv,\\
&  c_{\bv,\bu}\ge 0, f_{\bv,\bu}\ge 0, \quad \forall \bu, \bv.
\end{split}
\end{equation}
%Furthermore, let 
%\begin{equation}
%\begin{split}
%\theta_{\max}(\cN) =\log \max \ & \tr [J_{AB}^\cN R_{AB}] \\
%\text{s.t. } &  \sum_\bv b_{\bv}\le 1,\\
%& R_{AB}\ge 0, \ |W_{R_{AB}}(\bv,\bu)| \le b_{\bv}/d_B, \quad \forall \bu, \bv.
%\end{split}
%\end{equation}
This concludes the proof.
\end{proof}

\begin{corollary}[Max-thauma vs.~mana]
For a quantum channel $\cN_{A\to B}$, its max-thauma does not exceed its mana:
\begin{equation}
\theta_{\max}(\cN) \leq \cM(\cN).
\end{equation}
\end{corollary}

\begin{proof}
The proof is a direct consequence of the primal formulation in \eqref{eq:theta max channel min}. By setting $Y_{AB} = J^{\cN}_{AB}$, we find that
\begin{equation}
\theta_{\max}(\cN) \leq \log \min \{t : \max_{\bu}\sum_{\bv}|\tr [(A_A^{\bu}\ox A_B^{\bv})J_{AB}]|/d_B \le t\} = \cM(\cN),
\end{equation}
where the last equality follows from \eqref{eq:choi-mana}.
\end{proof}

 \begin{proposition}[Additivity]
 \label{prop:thauma-ch-additive}
For two given quantum channels $\cN_1$ and $\cN_2$, the max-thauma is additive in the following sense:
\begin{align}
    \theta_{\max}(\cN_1\ox\cN_2)=\theta_{\max}(\cN_1)+\theta_{\max}(\cN_2).    
\end{align}
 \end{proposition}
\begin{proof}
The idea of the proof is to utilize the primal and dual SDPs of $\theta_{\max}(\cN)$ from Proposition~\ref{prop:SDPs-max-thauma}. 
On the one hand, suppose that the optimal solutions to the primal SDPs for $\theta_{\max}(\cN_1)$ and $\theta_{\max}(\cN_2)$ are $\{R_1,b_{\bv_1}\}$ and $\{R_2,b_{\bv_2}\}$, respectively.
It is then easy to verify that $\{R_1\ox R_2,b_{\bv_1}b_{\bv_2}\}$ is a feasible solution to the SDP of $\theta_{\max}(\cN_1\ox\cN_2)$. Thus,
\begin{align}
    \theta_{\max}(\cN_1\ox\cN_2) \ge \log \tr [(J_{A_1B_1}^{\cN_1}\ox J_{A_2B_2}^{\cN_2})(R_1\ox R_2)]
    =\theta_{\max}(\cN_1)+\theta_{\max}(\cN_2).
\end{align}

On the other hand, considering Eq.~\eqref{eq:theta divergence}, suppose that the optimal solutions for $\cN_1$ and $\cN_2$ are $\cE_1$ and $\cE_2$, respectively. Noting that $\cM(\cE_1\ox\cE_2)=\cM(\cE_1)+\cM(\cE_2)\le 0$, and employing \eqref{eq:max-rel-ent-channels} and the additivity of the max-relative entropy, we find that
\begin{align}
    \theta_{\max}(\cN_1\ox\cN_2) & \le D_{\max}(\cN_1\ox\cN_2\|\cE_1\ox\cE_2)\\
    & =\theta_{\max}(\cN_1)+\theta_{\max}(\cN_2).
\end{align}
This concludes the proof.
\end{proof}

\bigskip
The following lemma is essential to establishing subadditivity of max-thauma of channels with respect to serial composition, as stated in Proposition~\ref{prop:thauma-ch-subadditive} below. We suspect that Lemma~\ref{lem:sub-add-max-ch-div} will find wide use in general resource theories beyond the magic resource theory considered in this paper. For example, it leads to an alternative proof of \cite[Proposition~17]{BHKW18}.

\begin{lemma}[Subadditivity of max-divergence of channels]
\label{lem:sub-add-max-ch-div}
Given completely positive maps $\cN^1_{A\to B}$, $\cN^2_{B \to C}$, $\cE^1_{A\to B}$, and $\cE^2_{B\to C}$, the following subadditivity inequality, with respect to serial compositions, holds for the max-channel divergence of \eqref{eq:max-div-channels-1}--\eqref{eq:max-rel-ent-channels}:
\begin{equation}
D_{\max}(\cN_2 \circ\cN_1 \| \cE_2 \circ \cE_1) \leq D_{\max}(\cN_1 \| \cE_1) + D_{\max}(\cN_2 \| \cE_2),
\end{equation}
where we have made the abbreviations $\cN_1\equiv \cN^1_{A\to B}$, $\cN_2 \equiv \cN^2_{B \to C}$, $\cE_1 \equiv \cE^1_{A\to B}$, and $\cE_2\equiv \cE^2_{B\to C}$.
\end{lemma}
\begin{proof}
Recall the ``data-processed triangle inequality'' from \cite{Christandl2017}:
\begin{equation}
D_{\max}(\cP(\rho)\| \omega) \leq D_{\max}(\rho\|\sigma) + D_{\max}(\cP(\sigma)\|\omega) \label{eq:DP-triangle},
\end{equation}
which holds for $\cP$ a positive map and $\rho$, $\omega$, and $\sigma$ positive semi-definite operators. Note that one can in fact see this as a consequence of the submultiplicativity of the operator norm and the data-processing inequality of max-relative entropy for positive maps:
\begin{align}
D_{\max}(\cP(\rho)\| \omega) & = 2 \log \left\| \omega^{-1/2} [\cP(\rho)]^{1/2} \right\|_\infty \\
 & = 2 \log \left\| \omega^{-1/2} [\cP(\sigma)]^{1/2} [\cP(\sigma)]^{-1/2} [\cP(\rho)]^{1/2} \right\|_\infty \\
& \leq  2 \log \left\| \omega^{-1/2} [\cP(\sigma)]^{1/2} \right\|_\infty \cdot \left\| [\cP(\sigma)]^{-1/2} [\cP(\rho)]^{1/2} \right\|_\infty \\
& = 2 \log \left\| \omega^{-1/2} [\cP(\sigma)]^{1/2} \right\|_\infty + 2 \log \left\| [\cP(\sigma)]^{-1/2} [\cP(\rho)]^{1/2} \right\|_\infty \\
& = D_{\max}(\cP(\rho)\|\cP(\sigma)) + D_{\max}(\cP(\sigma)\|\omega) \\
& \leq D_{\max}(\rho\|\sigma) + D_{\max}(\cP(\sigma)\|\omega) .
\end{align}
 Let us pick
\begin{equation}
\cP = \operatorname{id} \otimes\cN_2, \quad \rho = (\operatorname{id} \otimes \cN_1)(\Phi),
\quad \sigma = (\operatorname{id} \otimes \cE_1)(\Phi), \quad
\omega = (\operatorname{id} \otimes \cE_2)(\sigma),
\label{eq:DP-op-choices}
\end{equation}
where $\Phi$ denotes the maximally entangled state.
We find that
\begin{align}
& D_{\max}(\cN_2 \circ\cN_1 \| \cE_2 \circ \cE_1) \notag \\
& = D_{\max}((\operatorname{id} \otimes (\cN_2 \circ\cN_1))(\Phi)\| (\operatorname{id} \otimes (\cE_2 \circ \cE_1))(\Phi))  \\
& \leq D_{\max}((\operatorname{id} \otimes \cN_1)(\Phi)\|(\operatorname{id} \otimes \cE_1)(\Phi)) 
+ D_{\max}((\operatorname{id} \otimes (\cN_2 \circ \cE_1))(\Phi) \|(\operatorname{id} \otimes (\cE_2 \circ \cE_1))(\Phi)) \\
& \leq D_{\max}(\cN_1 \| \cE_1) + D_{\max}(\cN_2 \| \cE_2).
\end{align}
The first equality follows from \eqref{eq:max-rel-ent-channels}. The first inequality follows from \eqref{eq:DP-triangle} with the choices in \eqref{eq:DP-op-choices}.
The second inequality follows because $D_{\max}((\operatorname{id} \otimes \cN_1)(\Phi)\|(\operatorname{id} \otimes \cE_1)(\Phi)) = D_{\max}(\cN_1 \| \cE_1)$, as a consequence of \eqref{eq:max-rel-ent-channels}, and the channel divergence $D_{\max}(\cN_2 \| \cE_2)$ involves an optimization over all bipartite input states, one of which is $(\operatorname{id} \otimes \cE_1)(\Phi)$. 
\end{proof}

\begin{remark}
The  proof above applies to any divergence that obeys the data-processed triangle inequality, which includes the Hilbert $\alpha$-divergences of \cite{BG17}, as discussed in \cite[Appendix~A]{BHKW18}.
\end{remark}

 \begin{proposition}[Subadditivity]
 \label{prop:thauma-ch-subadditive}
For two given quantum channels $\cN_1$ and $\cN_2$, the max-thauma is subadditive in the following sense:
\begin{align}
    \theta_{\max}(\cN_2\circ\cN_1)\leq \theta_{\max}(\cN_1)+\theta_{\max}(\cN_2).    
\end{align}
 \end{proposition}
\begin{proof}
This is a direct consequence of Lemma~\ref{lem:sub-add-max-ch-div} above.
Let $\cE_i$ be the completely positive map satisfying $\cM(\cE_i)\leq 0$ and that is optimal for $\cN_i$ with respect to the max-thauma $\theta_{\max}$, for $i \in \{1,2\}$. Then applying Lemma~\ref{lem:sub-add-max-ch-div}, we find that
\begin{align}
 D_{\max}(\cN_2 \circ\cN_1 \| \cE_2 \circ \cE_1) & \leq D_{\max}(\cN_1 \| \cE_1) + D_{\max}(\cN_2 \| \cE_2) \\
&  = \theta_{\max}(\cN_1) + \theta_{\max}(\cN_2),
\end{align}
The equality follows from the assumption that $\cE_i$ is the completely positive map satisfying $\cM(\cE_i)\leq 0$, which is optimal for $\cN_i$ with respect to the max-thauma $\theta_{\max}$, for $i \in \{1,2\}$.

Given that, by assumption, $\cM(\cE_i)\leq 0$ for $i \in \{1,2\}$, it follows from Proposition~\ref{prop: M subadd}  that~$\cM\!\left(\cE_2 \circ \cE_1\right)\leq 0$. Since the max-thauma involves an optimization over all completely positive maps $\cE$ satsifying $\cM(\cE)\leq 0$, we conclude that
\begin{equation}
\theta_{\max}(\cN_2 \circ\cN_1)
\leq D_{\max}(\cN_2 \circ\cN_1 \| \cE_2 \circ \cE_1) \leq \theta_{\max}(\cN_1) + \theta_{\max}(\cN_2),
\end{equation}
which is the statement of the proposition.
\end{proof}

\begin{proposition}[Amortization inequality]\label{prop: theta amo}
For any quantum channel $\cN_{A \to B}$, the following inequality holds
\begin{align}
\sup_{\rho_A} [\theta_{\max}(\cN_{A\to B}(\rho_A))-\theta_{\max}(\rho_A)]\le 
\theta_{\max}(\cN),
\label{eq:amortized-thauma-no-ref}
\end{align}
with the optimization performed over input states $\rho_A$.
Moreover, the following inequality also holds 
\begin{align}
\sup_{\rho_{RA}} [\theta_{\max}((\operatorname{id}_R\ox \cN)(\rho_{RA}))-\theta_{\max}(\rho_{RA})]\le \theta_{\max}(\cN_{A\to B}).
\label{eq:amortized-thauma-with-ref}
\end{align}
\end{proposition}

 \begin{proof}
 The inequality in \eqref{eq:amortized-thauma-no-ref} is a direct consequence of reduction to states (Proposition~\ref{prop:max-thauma-reduction-to-states}) and subadditivity of max-thauma with respect to serial compositions (Proposition~\ref{prop:thauma-ch-subadditive}). Indeed, letting $\cN'$ be a replacer channel that prepares the state $\rho_A$, we find that
\begin{equation}
\theta_{\max}(\cN(\rho_A)) = \theta_{\max}(\cN\circ \cN') \leq \theta_{\max}(\cN) + \theta_{\max}(\cN') = \theta_{\max}(\cN) + \theta_{\max}(\rho_A). \label{eq:amortize-max-thauma-no-ref-no-opt}
\end{equation}
for all input states $\rho_A$,
from which we conclude \eqref{eq:amortized-thauma-no-ref}.

To arrive at the inequality in \eqref{eq:amortized-thauma-with-ref}, we make the substitution $\cN \to \operatorname{id} \otimes \cN$, apply the above reasoning, the additivity in Proposition~\ref{prop:thauma-ch-additive}, and the fact that the identity channel is free (CPWP), to conclude that the following  holds for all input states $\rho_{RA}$
\begin{align}
\theta_{\max}((\operatorname{id}_R \otimes \cN_{A\to B})(\rho_{RA}))-\theta_{\max}(\rho_{RA}) \le 
\theta_{\max}(\operatorname{id} \otimes \cN) = \theta_{\max}(\operatorname{id} ) + \theta_{\max}( \cN)
= \theta_{\max}( \cN),
\end{align}
from which we conclude \eqref{eq:amortized-thauma-with-ref}.
 \end{proof}

\bigskip

To summarize, the properties of $\theta_{\max}(\cN)$ are as follows:
\begin{enumerate}
\item Reduction to states: $\theta_{\max}(\cN) = \theta_{\max}(\sigma)$ when the channel $\cN$ is a replacer channel, acting as $\cN(\rho) = \tr[\rho]\sigma$ for an arbitrary input state $\rho$, where $\sigma$ is a state.
\item Monotonicity of $\theta_{\max}(\cN)$ under CPWP superchannels (including completely stabilizer-preserving superchannels).
\item Additivity under tensor products of channels: $\theta_{\max}(\cN_1\ox\cN_2)=\theta_{\max}(\cN_1)+\theta_{\max}(\cN_2)$.
\item Subadditivity under serial composition of channels: $\theta_{\max}(\cN_2\circ\cN_1)\leq \theta_{\max}(\cN_1)+\theta_{\max}(\cN_2)$.
\item  Faithfulness: $\cN \in \rm{CPWP}$ if and only if $\theta_{\max}(\cN)=0$.
\item Amortization inequality: $\sup_{\rho_{RA}} \theta_{\max}((\operatorname{id}_R\ox \cN)(\rho_{RA}))-\theta_{\max}(\rho_{RA})\le \theta_{\max}(\cN)$.
\end{enumerate}

\begin{remark}
Due to the subadditivity inequality in Proposition~\ref{prop:thauma-ch-subadditive}, the additivity identity in Proposition~\ref{prop:thauma-ch-additive}, and faithfulness in \eqref{eq:max-thauma-faithful}, the following identities hold
\begin{align}
\theta_{\max}(\cN_1) & = \sup_{\cN_2} \theta_{\max}([\operatorname{id} \ox \cN_1]\circ \cN_2) - \theta_{\max}( \cN_2),  \\
\theta_{\max}(\cN_1) & = \sup_{\cN_2} \theta_{\max}( \cN_2 \circ [\operatorname{id} \ox \cN_1]) - \theta_{\max}( \cN_2),  \\
\theta_{\max}(\cN_1) & = \sup_{\cN_2, \cN_3} \theta_{\max}( \cN_2 \circ [\operatorname{id} \ox \cN_1] \circ \cN_3) - \theta_{\max}( \cN_2) - \theta_{\max}( \cN_3),
\end{align}
which have the interpretation that amortization in terms of arbitrary pre- and post-processing does not increase the max-thauma of a quantum channel.
\end{remark}

%%%%%%%%%%%%%%%%%%%%%%%%%%%%%%%%%%%%%%%%%%%%
\section{Distilling magic from quantum channels}

\label{sec:magic gen}

\subsection{Amortized magic}

Since many physical tasks relate to quantum channels and time evolution rather than directly to quantum states, it is of interest to consider the non-stabilizer properties of quantum channels. Now having established suitable measures to quantify the magic of
quantum channels, it is natural to figure out the ability of
a quantum channel to generate magic from input quantum states. Let us begin by defining the amortized magic of a quantum channel:

 \begin{definition}[Amortized magic]
The amortized magic of a quantum channel $\cN_{A\to B}$ is defined relative to a magic measure $m(\cdot)$ via the following formula:
\begin{align}
m^{\mathcal{A}}(\cN) : = \sup_{\rho_{RA}} m((\operatorname{id}_R \otimes \cN)(\rho_{RA}))-m(\rho_{RA}).
 \end{align}
 The strict amortized magic of a quantum channel is defined as
 \begin{align}
\widetilde m^{\mathcal{A}}(\cN) : = \sup_{\rho_{RA}\in \STAB} m((\operatorname{id}_R \otimes \cN)(\rho_{RA})).
 \end{align}
\end{definition}

That is, the amortized magic is  defined as the largest
increase in magic that a quantum channel can realize after it acts on an arbitrary input quantum state. The strict amortized magic  is defined by finding the largest amount of magic that a quantum channel can realize when a stabilizer state is given to it as an input. Such amortized measures of resourcefulness of quantum channels were previously studied in the resource theories of quantum coherence (e.g., \cite{Mani2015,Garcia-Diaz2016,BenDana2017,Diaz2018}) and quantum entanglement (e.g., \cite{BHLS03,LHL03,Takeoka2014,KW17a,WW18}). They have been considered in the context of an arbitrary resource theory in \cite[Section~7]{KW17a}.

\begin{proposition}
Given a quantum channel $\cN_{A\to B}$, the following inequalities hold
\begin{align}
    \cM^{\mathcal{A}}(\cN) & \coloneqq\sup_{\rho_{RA}} \cM((\operatorname{id}_R \otimes \cN_{A\to B})(\rho_{RA}))-\cM(\rho_{RA})
    \le \cM(\cN),\\
    \theta_{\max}^{\mathcal{A}}(\cN) & \coloneqq\sup_{\rho_{RA}} \theta_{\max}((\operatorname{id}_R \otimes \cN_{A\to B})(\rho_{RA}))-\theta_{\max}(\rho_{RA})
    \le \theta_{\max}(\cN).
\end{align}
\end{proposition}
\begin{proof}
These statements are an immediate consequence of the amortization inequality for $\cM(\rho)$ and $\theta_{\max}(\rho)$ given in Propositions~\ref{prop: M amo} and \ref{prop: theta amo}, respectively.
\end{proof}

\subsection{Distillable magic of a quantum channel}

The most general protocol for distilling some resource by means of a quantum channel $\cN$ employs $n$ invocations of the channel $\cN$ interleaved by free channels \cite[Section~7]{KW17a}. In our case, the resource of interest is magic, and here we take the free channels to be the CPWP channels discussed in Section~\ref{sec:PWP channel}. In such a protocol, the instances of the channel $\cN$ are invoked one at a time, and we can integrate all CPWP channels between one use of $\cN$ and the next into a single CPWP channel, since the CPWP channels are closed under composition. The goal of such a protocol is to distill magic states from the channel.

In more detail, the most general protocol for distilling magic from a quantum channel proceeds as follows: one starts by preparing the systems $R_1 A_1$ in a state $\rho^{(1)}_{R_1 A_1}$ with non-negative Wigner function, by employing a free CPWP channel $\cF^{(1)}_{\emptyset \to R_1 A_1}$, then applies the channel $\cN_{A_1 \to B_1}$, followed by a CPWP channel $\cF^{(2)}_{R_1B_1 \to R_2 A_2}$, resulting in the state
\begin{equation}
\rho^{(2)}_{R_2 A_2} \coloneqq \cF^{(2)}_{R_1B_1 \to R_2 A_2}(( \operatorname{id}_{R_1} \ox\cN_{A_1\to B_1})(\rho^{(1)}_{R_1 A_1})).
\end{equation}
Continuing the above steps, given state $\rho^{(i)}_{R_{i} A_{i}}$ after the action of $i-1$ invocations of the channel $\cN_{A \to B}$ and interleaved CPWP channels, we apply the channel $\cN_{A_{i}\to B_{i}}$ and the CPWP channel $\cF^{(i+1)}_{R_{i}B_{i}\to R_{i+1} A_{i+1}}$, obtaining the state
\begin{equation}
\rho^{(i+1)}_{R_{i+1} A_{i+1}} \coloneqq \cF^{(i+1)}_{R_{i}B_{i} \to R_{i+1} A_{i+1}}(( \operatorname{id}_{R_i} \ox\cN_{A_{i}\to B_{i}})(\rho^{(i)}_{R_i A_i})).
\end{equation}
After $n$ invocations of the channel $\cN_{A \to B}$ have been made, the final free CPWP channel $\cF^{(n+1)}_{R_{n}B_{n}\to S}$ produces a state $\omega_S$ on system $S$, defined as
\begin{equation}
\omega_S \coloneqq \cF^{(n+1)}_{R_{n}B_{n}\to S}(\rho^{(n)}_{R_{n} A_{n}}).
\end{equation}
Such a protocol is depicted in Figure~\ref{fig:magic-distill-channel}.

\begin{figure}
\begin{center}
\includegraphics[width=\textwidth]{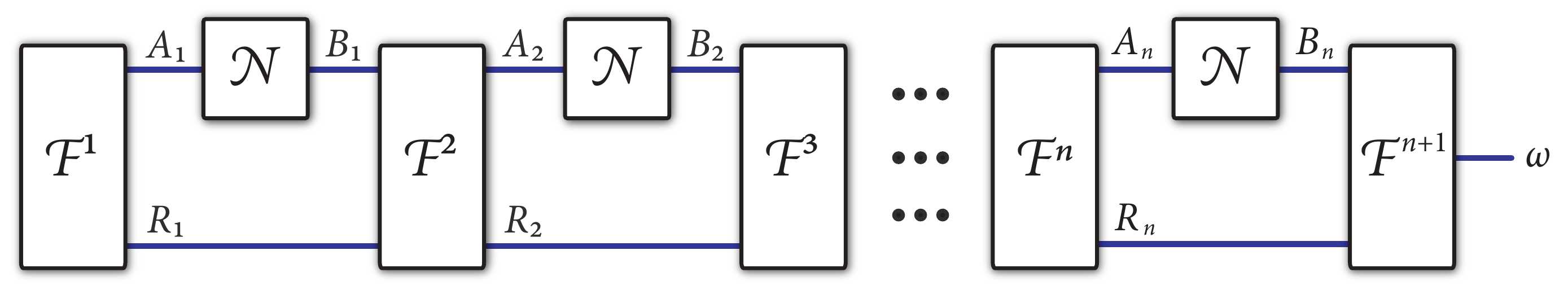}
\caption{The most general protocol for distilling magic from a quantum channel.}
\label{fig:magic-distill-channel}
\end{center}

\end{figure}

Fix $\ve \in [0,1]$ and $k\in \mathbb{N}$. The above procedure is an $(n,k,\ve)$ $\psi$-magic distillation protocol with rate $k/n$ and error $\ve$, if the state $\omega_S$ has a high fidelity with $k$ copies of the target magic state $\psi$,  
\begin{align}
   \bra{\psi}^{\ox {k}}\omega_S \ket {\psi}^{\ox {k}} \ge 1-\ve.
\end{align}
A rate $R$ is achievable for $\psi$-magic state distillation from the channel $\cN$, if for all $\ve\in(0,1]$, $\delta > 0$, and sufficiently large $n$, there exists an $(n, n(R-\delta), \ve)$ $\psi$-magic state distillation protocol of the above form. The $\psi$-distillable magic of the channel $\cN$ is defined to be the supremum of all achievable rates and is denoted by $C_\psi(\cN)$.

A common choice for a non-Clifford gate is the $T$-gate. The qutrit $T$ gate \cite{Howard2012} is given by
\begin{align}
T=\left(\begin{matrix}
\xi & 0 & 0\\
0 & 1 & 0\\
0& 0 & \xi^{-1}
\end{matrix}
\right),
\end{align}
where $\xi=e^{2\pi i/9}$ is a primitive ninth root of unity. The $T$ gate leads to the $T$ magic state
\begin{equation}
\ket{T} \coloneqq T \ket{+},
\end{equation}
by inputting the stabilizer state $\ket{+}$ to the $T$ gate. Furthermore, by the method of state injection \cite{Gottesman1999,Zhou2000}, one can generate a $T$ gate by acting with stabilizer operations on the $T$ state $\ket{T}$.

In what follows, we use quantum hypothesis testing to establish an upper bound on the rate at which one can distill qutrit $T$ states. The proof follows the general method in \cite[Theorem~1]{BHLS03} and \cite[Theorem~1]{BenDana2017}, which was later generalized to an arbitrary resource theory in \cite[Section~7]{KW17a}.

\begin{proposition}
Given a quantum channel $\cN$, the following upper bound holds for the rate $R=k/n$ of an $(n,k,\ve)$ $T$-magic distillation protocol:
\begin{align}
     R \le \frac{1}{\log(1+2\sin(\pi/18))}\left(\theta_{\max}(\cN)+\frac{\log(1/[1-\ve])}{n}\right).
\end{align}
Consequently, the following upper bound holds for the $T$-distillable magic of a quantum channel $\cN$:
\begin{align}
    C_T(\cN) \le \frac{\theta_{\max}(\cN)}{\log (1+2\sin (\pi/18))}.
\end{align}
\end{proposition}
\begin{proof}
Consider an arbitrary  $(n,k,\ve)$ $T$-magic state distillation protocol of the form described previously.
Such a protocol uses the channel $n$ times, starting from the state $\rho^{(1)}_{R_1 A_1}$ with non-negative Wigner function and generating $\rho^{(2)}_{R_2 A_2},\ldots,\rho^{(n)}_{R_n A_n},$ and $ \omega_S $ step by step along the way, such that the final state $\omega_S$ has fidelity $1-\ve$ with $\ket T^{\ox {k}}$, where $\ket T=T \ket +$ is the corresponding magic state of the $T$ gate. By assumption,
it follows that
\begin{align}
    \tr [\proj T^{\ox k}\omega_S] \ge 1-\ve,
\end{align}
while the result in \cite{Wang2018} implies that
\begin{align}
    \tr [\proj T^{\ox k}\sigma_S] \le (1+2\sin(\pi/18))^{-k}
\end{align}
for all $\sigma_S\in\cW$ with the same dimension as $\omega_S$. Applying the data processing inequality for the max-relative entropy, with respect to the measurement channel
\begin{equation}
(\cdot) \to \tr[\proj T^{\otimes k}(\cdot)] \proj 0 + \tr[(I^{\otimes k} - \proj T^{\otimes k})(\cdot)] \proj 1,
\end{equation}
 we find that
\begin{align}
    {\theta_{\max}(\omega_S)}&\ge \log[(1-\ve)(1+2\sin(\pi/18))^{k}]\\
    & \ge \log(1-\ve)+k\log(1+2\sin(\pi/18)).
\end{align}

Moreover, by labeling $\omega_S$ as $\rho^{(n+1)}$, we find that
\begin{align}
    \theta_{\max}(\rho^{(n+1)}) &= \sum_{j=1}^n [\theta_{\max}(\rho^{(j+1)}) -\theta_{\max}(\rho^{(j)})] \\
    &=\sum_{j=1}^n [\theta_{\max}((\cF^{(j+1)}\circ[\operatorname{id}\ox \cN])(\rho^{(j)})) -\theta_{\max}(\rho^{(j)})]\\
    &\le\sum_{j=1}^n [\theta_{\max}([\operatorname{id}\ox \cN](\rho^{(j)})) -\theta_{\max}(\rho^{(j)})]\\
    &\le n\theta_{\max}(\cN).
\end{align}
The first equality follows because $\theta_{\max}(\rho^{(1)}) = 0$ and by adding and subtracting terms. The first inequality follows because the max-thauma of a state does not increase under the action of a CPWP channel. The last inequality follows from applying Proposition~\ref{prop: theta amo}.

%Moreover, by considering the first free channel $\cF^{(1)}$ to be equivalent to a channel that traces out its input and replaces with the initial state $\rho^{(1)}_{R_1 A_1}$, we find that
%\begin{align}
%    \theta_{\max}(\omega_S) &= \theta_{\max} (\cF^{(1)}\circ \cN \circ \cF^{(2)} \circ \cN \circ  \cdots \circ \cN \circ \cF^{(n+1)}) \\
%    & \leq \sum_{i=1}^{n+1} \theta_{\max}(\cF^{(i)}) +  n\theta_{\max}(\cN) \\
%    &= n\theta_{\max}(\cN).
%\end{align}
%The first equality follows from \eqref{prop:max-thauma-reduction-to-states}.
%The first inequality follows from iteratively applying the subadditivity inequality in Proposition~\ref{prop:thauma-ch-subadditive}, and the second inequality follows because each channel $\cF^{(i)}$ is CPWP and thus $\theta_{\max}(\cF^{(i)})=0$ according to Proposition~\ref{eq:max-thauma-faithful}.

Hence,
\begin{align}
    n\theta_{\max}(\cN) \ge  \log(1-\ve)+k\log(1+2\sin(\pi/18)),
\end{align}
which implies that
\begin{align}
    R = k/n\le \frac{1}{\log(1+2\sin(\pi/18))}\left(\theta_{\max}(\cN)+\frac{\log(1/[1-\ve])}{n}\right).
\end{align}
This concludes the proof.
\end{proof}

\bigskip
We note here that one could also use the subadditivity inequality in Proposition~\ref{prop:thauma-ch-subadditive} to establish the above result.
We further note here that similar results in terms of max-relative entropies have been found in the context of other resource theories. Namely, a channel's max-relative entropy of 
entanglement is an upper bound on its distillable secret key when assisted by LOCC channels \cite{Christandl2017}, the max-Rains information of a quantum channel is an upper bound on its distillable entanglement when assisted by completely PPT preserving channels \cite{BW17}, and the max-$k$-unextendibility of a quantum channel is an upper bound on its distillable entanglement when assisted by $k$-extendible channels \cite{KDWW18}.

\subsection{Injectable quantum channel}

In any resource theory of quantum channels, it tends to simplify for those channels that can be implemented by the action of a free channel on the tensor product of the channel input state and a resourceful state \cite[Section~7]{KW17a} and \cite[Section~6]{Wilde2018}. The situation is no different for the resource theory of magic channels. In fact, particular channels with the aforementioned structure have been considered for a long time in the context of magic states, via the method of state injection \cite{Gottesman1999,Zhou2000}. Here we formally define an injectable channel as follows:

\begin{definition}[Injectable channel]
A quantum channel $\cN$ is called injectable with associated resource state $\omega_C$ if there exists  a CPWP channel $\Lambda_{AC \to B}$ such that the following equality holds for all input states~$\rho_A$:
\begin{align}
	\cN_{A\to B}(\rho_A)=\Lambda_{AC\to B}(\rho_A\ox \omega_{C}).
\end{align}
\end{definition}

The notion of a resource-seizable channel was introduced in \cite{BHKW18,Wilde2018}, and here we consider the application of this notion in the context of magic resource theory:

\begin{definition}[Resource-seizable channel]
Let $\cN_{A \to B}$ be an injectable channel with associated resource state $\omega_C$. The channel $\cN$ is resource-seizable if there exists a free state $\kappa^{\operatorname{pre}}_{RA}$ with non-negative Wigner function and a post-processing free CPWP channel $\cF^{\operatorname{post}}_{RB\to C}$ such that
\begin{equation}
\cF^{\operatorname{post}}_{RB\to C}(\cN_{A \to B}(\kappa^{\operatorname{pre}}_{RA})) = \omega_C.
\end{equation}
In the above sense, one seizes the resource state $\omega_C$ by employing free pre- and post-processing of the channel~$\cN_{A \to B}$.
\end{definition}

An interesting and prominent example of an injectable channel that is also resource seizable is the channel $\cT$ corresponding to the $T$ gate. This channel $\cT$  has the following action $\cT(\rho)\coloneqq T \rho T^\dag$ on an input state $\rho$.
This channel is injectable with associated resource state  $\omega_C = \proj T$, since one can use the method of circuit injection \cite{Zhou2000} to obtain the channel $\cT$ by acting on $\proj T$ with stabilizer operations. It is resource seizable because one can act on the free state $\proj +$ with the channel $\cT$ in order to seize the underlying resource state $\proj T = \cT(\proj + )$.

As a generalization of the $\cT$ channel example above, consider the channel $\Delta^{\mathbf{p}} \circ T$, where $\Delta^{\mathbf{p}}$ is a dephasing channel of the form
\begin{equation}
\Delta^{\mathbf{p}}(\rho) = p_0 \rho + p_1 Z \rho Z^\dag + p_2 Z^2 \rho (Z^2)^\dag
\end{equation}
where $\mathbf{p} = (p_0, p_1, p_2)$, $p_0,p_1,p_2\geq 0$, and $p_0 + p_1 + p_2 = 1$. The channel is injectable with resource state $\Delta^{\mathbf{p}}(\proj T)$, because the same method of circuit injection leads to the channel $\Delta^{\mathbf{p}} \circ T$ when acting on the resource state $\Delta^{\mathbf{p}}(\proj T)$. Furthermore, the channel $\Delta^{\mathbf{p}} \circ T$ is resource seizable because one recovers the resource state $\Delta^{\mathbf{p}}(\proj T)$ by acting with $\Delta^{\mathbf{p}} \circ T$ on the free state $\proj +$.

For such injectable channels, the resource theory of magic channels simplifies in the following sense:
\begin{proposition}
\label{prop:injectable-simplify}
Let $\cN$ be an injectable channel with associated resource state $\omega_C$. Then the following inequalities hold
\begin{equation}
\cM(\cN)  \leq \cM(\omega_C), \qquad
\boldsymbol{\theta}(\cN)  \leq \boldsymbol{\theta}(\omega_C), \label{eq:mana-inj-1}
\end{equation}
where $\boldsymbol{\theta}$ denotes the generalized thauma measures from Section~\ref{sec:gen-thauma}.
If $\cN$ is also resource seizable, then the following equalities hold
\begin{equation}
\cM(\cN)  = \cM(\omega_C), \qquad
\boldsymbol{\theta}(\cN)  = \boldsymbol{\theta}(\omega_C). \label{eq:mana-seizable-1}
\end{equation}
\end{proposition}

\begin{proof}
We first prove the first inequality in \eqref{eq:mana-inj-1}. Consider that
\begin{align}
\cM(\cN) & = \log \max_{\bu} \| \cN_{A \to B}(A^{\bu}_A)\|_{W,1} \\
& = \log \max_{\bu} \| \Lambda_{AC \to B}(A^{\bu}_A \otimes \omega_C)\|_{W,1} \\
& \leq \log \max_{\bu} \| A^{\bu}_A \otimes \omega_C\|_{W,1} \\
& = \log \max_{\bu} \| A^{\bu}_A \|_{W,1} + \log \| \omega_C \|_{W,1} \\
& = \log \| \omega_C \|_{W,1} \\
& = \cM(\omega_C).
\end{align}
The first two equalities follow from definitions. The inequality follows from Lemma~\ref{lem:DP-Wigner-tnorm} in the appendix. The third equality follows because the Wigner trace norm is multiplicative for tensor-product operators. The fourth equality follows because $\| A^{\bu}_A \|_{W,1} = 1$ for any phase-space point operator $A^{\bu}$. 

We now prove the second inequality in \eqref{eq:mana-inj-1}:
\begin{align}
\boldsymbol{\theta}(\cN) & = \inf_{\cE : \cM(\cE)\leq 0}\sup_{\psi_{RA}}
\mathbf{D}( \cN_{A \to B}(\psi_{RA})\|\cE_{A \to B}(\psi_{RA})) \\
&  = \inf_{\cE : \cM(\cE)\leq 0}\sup_{\psi_{RA}}
\mathbf{D}( \Lambda_{AC \to B}(\psi_{RA} \otimes \omega_C)\|\cE_{A \to B}(\psi_{RA})) \\
& \leq \inf_{\sigma_C\geq 0 : \cM(\sigma_C)\leq 0}\sup_{\psi_{RA}}
\mathbf{D}( \Lambda_{AC \to B}(\psi_{RA} \otimes \omega_C)\|\Lambda_{AC \to B}(\psi_{RA} \otimes \sigma_C)) \\
& \leq \inf_{\sigma_C\geq 0 : \cM(\sigma_C)\leq 0}\sup_{\psi_{RA}}
\mathbf{D}( \psi_{RA} \otimes \omega_C \|\psi_{RA} \otimes \sigma_C) \\
& = \inf_{\sigma_C\geq 0 : \cM(\sigma_C)\leq 0}
\mathbf{D}(  \omega_C \| \sigma_C) \\
& = \boldsymbol{\theta}(\omega_C).
\end{align}
The first two equalities follow from definitions. The first inequality follows because the completely positive map $\cE = \Lambda_{AC\to B}(\cdot \otimes \sigma_C)$ with $\sigma_C \in \cW$ is a special kind of completely positive map such that $\cM(\cE)\leq 0$, due to the first inequality in \eqref{eq:mana-inj-1}. The second inequality follows from data processing under the channel $\Lambda_{AC \to B}$. The third equality follows because the generalized divergence is invariant under tensoring its two arguments with the same state $\psi_{RA}$ (again a consequence of data processing \cite{Wilde2014a}). The final equality follows from the definition in \eqref{eq:GT-state}.

The inequalities in \eqref{eq:mana-seizable-1} are a direct consequence of the definition of a resource-seizable channel, the fact that both the mana and the generalized thauma are monotone under the action of a CPWP superchannel (Theorems~\ref{prop: M monotone} and \ref{thm:gen-thauma-monotone-PWP-superch}, respectively), and with $\cF^{\operatorname{post}}_{RB\to C}(\cN_{A \to B}(\kappa^{\operatorname{pre}}_{RA}))$ understood as a particular kind of superchannel that manipulates $\cN_{A \to B}$ to the state $\omega_C$. Furthermore, it is the case that the channel measures reduce to the state measures when evaluated for preparation channels that take as input a trivial one-dimensional system, for which the only possible ``state'' is the number one, and output a state on the output system (see Proposition~\ref{prop:mana-reduction-to-states} and \eqref{prop:max-thauma-reduction-to-states}).
\end{proof}

\bigskip

Applying Proposition~\ref{prop:injectable-simplify} to the channel $\cT$ and applying some of the results in \cite{Wang2018}, we find that
\begin{align}
\theta_{\max}(\cT) = \theta(\cT) = \theta_{\max}(\proj T) = \theta(\proj T) =  \log (1+2\sin(\pi/18)).
 \end{align}

The notion of an injectable channel also improves the upper bounds  on the distillable magic of a quantum channel:
\begin{proposition}
Given an injectable quantum channel $\cN$ with associated resource state $\omega_C$, the following upper bound holds for the rate $R=k/n$ of an $(n,k,\ve)$ $T$-magic distillation protocol:
\begin{align}
     R  \le \frac{1}{\log(1+2\sin(\pi/18))(1-\ve)}\left(\theta(\omega_C)+\frac{h_2(\ve)}{n}\right),
     \label{eq:weak-conv-bnd-1}
\end{align}
where $h_2(\ve) \coloneqq -\ve \log_2 \ve - (1-\ve)\log_2(1-\ve)$.
Consequently, the following upper bound holds for the $T$-distillable magic of the injectable quantum channel~$\cN$:
\begin{align}
    C_T(\cN) \le \frac{\theta(\omega_C)}{\log (1+2\sin (\pi/18))}.
    \label{eq:inj-ch-dist-magic-bnd}
\end{align}
\end{proposition}
\begin{proof}
Consider an arbitrary $(n,k,\ve)$ $T$-magic state distillation protocol of the form described previously. Due to the injection property, it follows that such a protocol is equivalent to a CPWP channel acting on the resource state $\omega_C^{\otimes n}$ (see Figure~5 of \cite{KW17a}). So the channel distillation problem reduces to a state distillation problem.
Applying Proposition~4 of \cite{Wang2018} and standard inequalities for the hypothesis testing relative entropy from \cite{Wang2012}, we conclude the bound in \eqref{eq:weak-conv-bnd-1}.
Then taking limits, we arrive at \eqref{eq:inj-ch-dist-magic-bnd}.
\end{proof}

%%%%%%%%%%%%%%%%%%%%%%%%%%%%%%%%%%
\section{Magic cost of a quantum channel}\label{sec: magic cost}
\subsection{Magic cost of exact channel simulation}
Beyond magic distillation via quantum channels, the magic measures of quantum channels can also help us investigate the magic cost in quantum gate synthesis. In the past two decades, tremendous progress has been accomplished in the area of gate synthesis for qubits (e.g., \cite{Bremner2002,Amy2014,Jones2013a,Selinger2012,Ross2016,Gosset2014,Bocharov2015,Wiebe2016a}) and qudits (e.g., \cite{Muthukrishnan2000,Brennen2006,Bullock2004,Di2013,Di2015}).  Elementary two-qudit gates include the controlled-increment gate \cite{Brennen2006} and the generalized controlled-$X$ gate \cite{Di2013,Di2015}. More recently, the synthesis of single-qutrit gates was studied in \cite{Prakash2018,Glaudell2018}.

Of particular interest is to study exact gate synthesis of multi-qudit unitary gates from elements of the Clifford group supplemented by $T$ gates. More generally, a fundamental question is to determine how many instances of a given quantum channel $\mathcal{N}'$ are required simulate another quantum channel $\mathcal{N}$, when supplemented with CPWP channels. That is, such a channel synthesis protocol has the following form:
\begin{equation}
\cN_{A \to B} = \cF^{n+1}_{R_n B_n'\to B} \circ \cN'_{A_n' \to B_n'} \circ \cF^{n}_{R_{n-1} B_{n-1}'\to R_{n} A_{n}'} \circ \cdots \circ
\cF^{2}_{R_1 B_1'\to R_2 A_2'} \circ \cN'_{A_1' \to B_1'} \circ \cF^{1}_{A \to R_{1} A_{1}'},
\label{eq:ch-simulation-syn}
\end{equation}
as depicted in Figure~\ref{fig:synthesis-channel}.
Let $S_{\cN'}(\cN)$ denote the smallest number of $\cN'$ channels required to implement the quantum channel $\cN$ exactly.
Note that it might not always be possible to have an \textit{exact} simulation of the channel $\cN$ when starting from another channel $\cN'$. For example, if $\cN$ is a unitary channel and $\cN'$ is a noisy depolarizing channel, then this is not possible. In this case, we define $S_{\cN'}(\cN)=\infty$.

\begin{figure}
\begin{center}
\includegraphics[width=\textwidth]{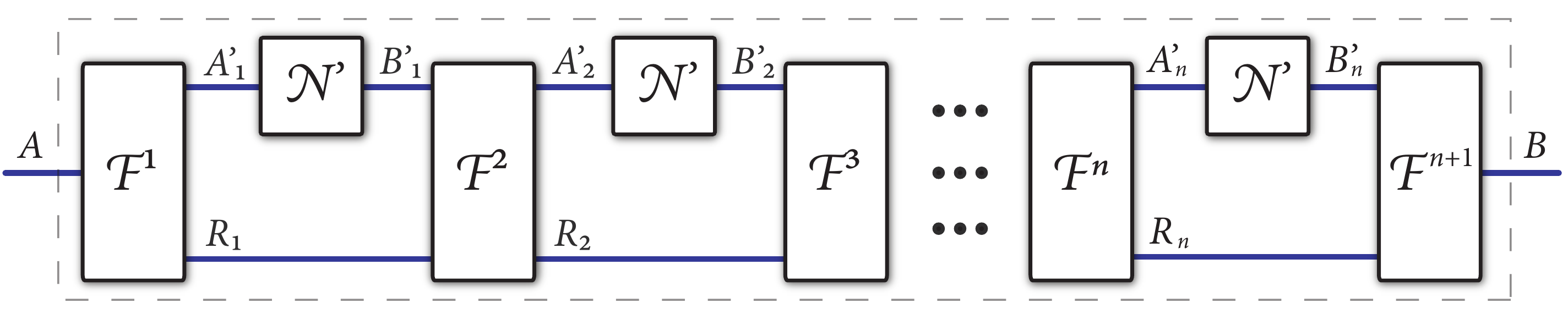}
\caption{The most general protocol for exact synthesis of a channel $\cN_{A\to B}$ starting from $n$ uses of another quantum channel $\cN'$, along with free CPWP channels $\cF^i$, for $i\in \{1,\ldots,n\}$.}
\label{fig:synthesis-channel}
\end{center}

\end{figure}

In the following, we establish lower bounds on gate synthesis by employing the channel measures of magic introduced previously.
\begin{proposition}
\label{prop:lower-bnd-ch-syn}
For any qudit quantum channel $\cN$, the number of channels $\cN'$ required to implement it is bounded from below as follows:
\begin{align}
S_{\cN'}(\cN) \ge \max\left\{\frac{\cM(\cN)}{\cM(\cN')}, \frac{\theta_{\max}(\cN)}{\theta_{\max}(\cN')}\right\}.
\end{align}
If the channel $\cN$ is injectable with associated resource state $\omega_C$, then the following bound holds
\begin{align}
S_{\cN'}(\cN) \ge \max\left\{\frac{\cM(\cN)}{\cM(\omega_C)}, \frac{\theta_{\max}(\cN)}{\theta_{\max}(\omega_C)}\right\}.
\label{eq:inj-statement-synthesis}
\end{align}
\end{proposition}

\begin{proof}
Suppose that the simulation of $\cN$ is realized as in \eqref{eq:ch-simulation-syn}. Applying Proposition~\ref{prop: M subadd} iteratively, we find that
\begin{align}
\cM(\cN) & \leq n \cM(\cN') + \sum_{i=1}^{n+1} \cM(\cF^i)  
  = n \cM(\cN')  ,
\end{align}
where the equality follows from Proposition~\ref{prop:faithfulness-mana} and the assumption that each $\cF^i$ is a CPWP channel. Then $n \geq \frac{\cM(\cN)}{\cM(\cN')}$. Since this inequality holds for an arbitrary channel synthesis protocol, we find that $S_{\cN'}(\cN) \ge \frac{\cM(\cN)}{\cM(\cN')}$.

Applying Propositions~\ref{prop:thauma-ch-subadditive} and \ref{eq:faithful-gen-thauma} in a similar way, we conclude that $S_{\cN'}(\cN) \ge \frac{\theta_{\max}(\cN)}{\theta_{\max}(\cN')}$.

If the channel is injectable, then the upper bounds in \eqref{eq:mana-inj-1} apply, from which we conclude \eqref{eq:inj-statement-synthesis}.
\end{proof}

\bigskip
As a direct application, we investigate gate synthesis of elementary gates. In the following, we prove that four $T$ gates are necessary to synthesize a controlled-controlled-X qutrit gate (CCX gate) exactly.
\begin{proposition}
To implement a controlled-controlled-$X$ qutrit gate, at least four qutrit $T$ gates are required. 
\end{proposition}
\begin{proof}
By direct numerical evaluation,
we find that 
\begin{align}
    S_T(CCX) \ge \frac{\cM(CCX)}{\cM(\proj T)} \ge \frac{2.1876}{0.6657} \ge 3.2861,
\end{align}
which means that four qutrit $T$ gates are necessary to implement a qutrit $CCX$ gate. 
\end{proof}

\bigskip
For NISQ devices, it is natural to consider gate synthesis under realistic quantum noise. One common noise model in quantum information processing is the depolarizing channel:
\begin{equation}
{\cD_p}(\rho) = (1-p) \rho + \frac{p}{d^2-1} \sum_{\substack{0\leq i,j\leq d-1 \\ (i,j)\neq (0,0)}} X^i Z^j \rho (X^i Z^j)^\dagger.
\end{equation}

Suppose that a $T$ gate is not available, but instead only a noisy version $\cD_p\circ T$ of it is. Then it is reasonable to consider the number of noisy $T$ gates required to implement a low-noise CCX gate, and the resulting lower bound is depicted in Figure~\ref{fig: num T CCX}. Considering the depolarizing noise ($p=0.01$) and applying Proposition~\ref{prop:lower-bnd-ch-syn}, the lower bound is given by 
\begin{align}
\frac{\cM(\cD_{0.01}^{\ox 3}\circ CCX)}{\cM{(\cD_p\circ T)}}.
\end{align}

\begin{figure}
\begin{center}
\includegraphics[width=8.5cm]{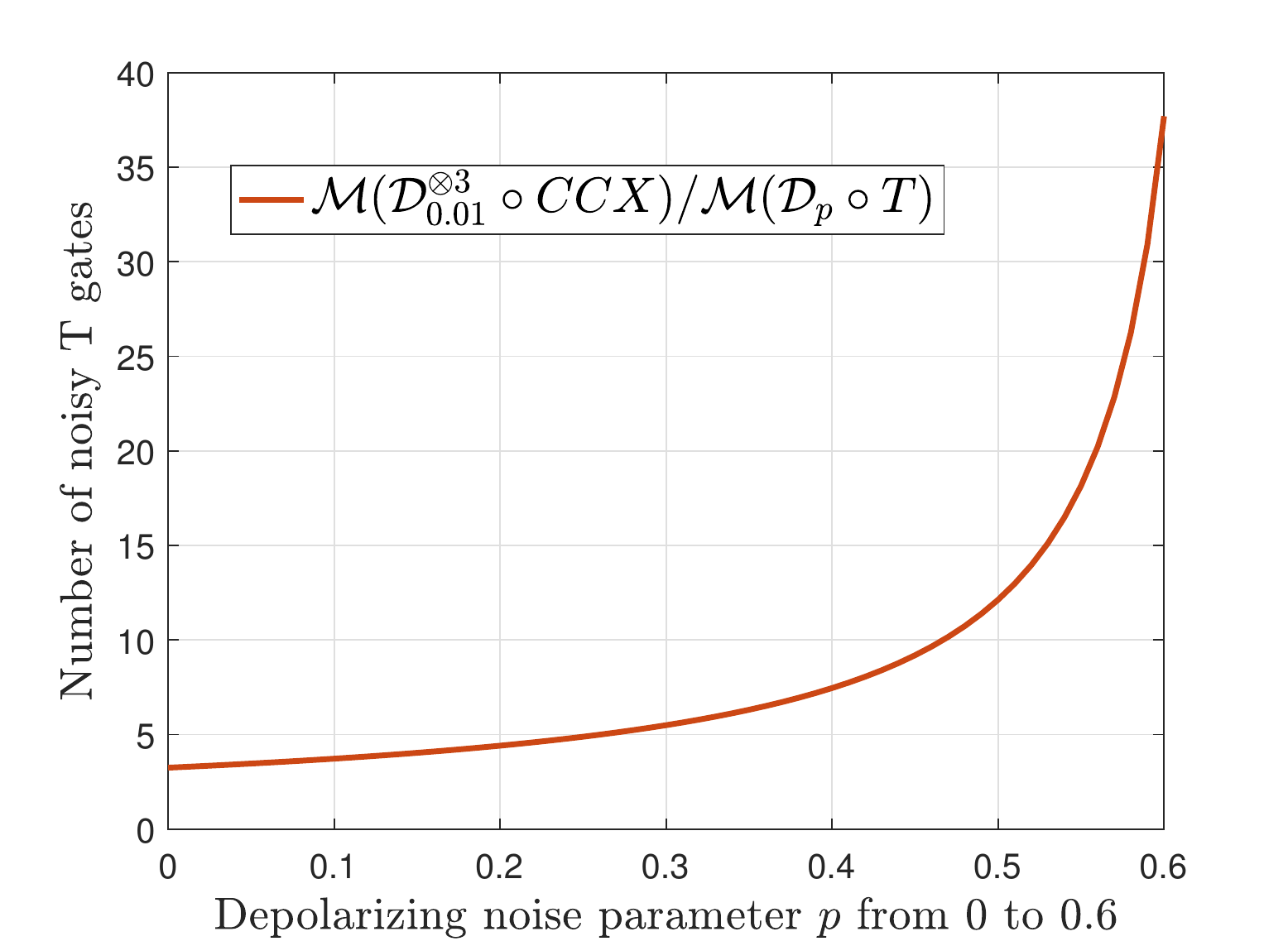}
\caption{Lower bound on the number of noisy $T$ gates required to implement a low-noise CCX gate.}
\label{fig: num T CCX}
\end{center}
\end{figure}

\subsection{Magic cost of approximate channel simulation}

Here we consider the magic cost of approximate channel simulation, which allows for a small error in the simulation process. To be specific, we establish the following proposition:
\begin{proposition}
For any qudit channel $\cN$ (with odd dimensions), the following lower bound for the number of channels $\cN'$ required to implement it with error tolerance $\ve$:
\begin{equation}
\begin{split}
S_{\cN'}^{\ve}(\cN) \ge  \min\{k: k\cM(\cN')\ge \cM(\widetilde\cN), \frac12\|\widetilde \cN - \cN\|_\di \leq \ve, \widetilde\cN\in \text{CPTP}\},
\end{split}
\end{equation}
where $\|\cdot\|_\di:=\sup_{k\in \mathbb N} \sup_{\|X\|_1 \leq 1}\|(\cdot\ox \id_k)(X)\|_1$ denotes the diamond norm.
\end{proposition}
To get this bound, we minimize the lower bound of exact magic cost in Proposition~\ref{prop:lower-bnd-ch-syn} over the quantum channels that are $\ve$-close to the target channel in terms of diamond norm.
Using the SDP form of the diamond norm \cite{Watrous2009} the above lower bound can be computed via the following SDP:
\begin{equation}
\begin{split}
 \log \min \ & t\\
 \text{s.t. } & t2^{\cM(\cN')}\ge \|\widetilde\cN(A_\bu)\|_{W,1}, \forall \bu,\\
  & \tr_B Y_{AB}\le \ve \1_A, Y\ge 0, Y\ge J_{\cN} - J_{\widetilde\cN},\\
  & J_{\widetilde\cN}\ge 0, \tr_B J_{\widetilde\cN}=\1_A.
\end{split}
\end{equation}
Moreover, one could also replace mana with thauma in the above resource estimation. Also, it is possible and interesting to exactly characterize the minimum error of channel simulation under CPWP bipartite channels. We leave these for future study.

\section{Classical simulation of quantum channels}

\label{sec: classical simulation}

\subsection{Classical algorithm for simulating noisy quantum circuits}

\label{sec: classical simulation mana}

An operational meaning associated with mana is that it quantifies the rate at which a quantum circuit can be simulated on a classical computer. Inspired by \cite{PWB15}, we propose an algorithm for simulating quantum circuits in which the operations can potentially be noisy. We show that the complexity of this algorithm scales with the mana (the logarithmic negativity) of quantum channels, establishing mana as a useful measure for measuring the cost of classical simulation of a (noisy) quantum circuit. For recent independent and related work, see \cite{RLCK19}.

Let $\mathcal{H}_d^{\otimes n}$ be the Hilbert space of an $n$-qudit system. Consider an evolution that consists of the sequence  $\{\mathcal{N}_l\}_{l=1}^L$ of channels acting on an input state $\rho$. Then the probability of  observing the POVM measurement outcome $E$, where $0\leq E \leq I$, can be computed according to the Born rule  as
\begin{equation}
	\tr\big[E(\mathcal{N}_L\circ\cdots\circ \mathcal{N}_1)(\rho)\big]
	=\sum_{\overrightarrow{\bu}}W(E|\bu_L)\prod_{l=1}^{L}W_{\mathcal{N}_l}(\bu_l|\bu_{l-1})W_\rho(\bu_0),
\end{equation}
where $\overrightarrow{\bu}=(\bu_0,\ldots,\bu_L)$ represents a vector in the discrete phase space and $W(E|\cdot)$ is the discrete Wigner function of the measurement operator $E$ (cf., Eq.~\eqref{eq:wig of measurement}). For the base case $L=1$, this follows from the properties of the discrete Wigner function:
\begin{align}
    \sum_{\bu,\bu'}W(E|\bu')W_{\mathcal{N}}(\bu'|\bu)W_\rho(\bu)
	&=\sum_{\bu,\bu'}\tr\!\big[EA_{\bu'}\big]
	\frac{1}{d}\tr\!\big[A_{\bu'}\mathcal{N}(A_{\bu})\big]
	\frac{1}{d}\tr\!\big[A_{\bu}\rho\big]\\
	&=\sum_{\bu'}\tr\!\big[EA_{\bu'}\big]
	\frac{1}{d}\tr\!\big[A_{\bu'}\mathcal{N}(\rho)\big]\\
	&=\tr\!\big[E\mathcal{N}(\rho)\big].
\end{align}
The case of general $L$ follows by induction.

Our goal is to estimate $\tr\big[E(\mathcal{N}_L\circ\cdots\circ \mathcal{N}_1)(\rho)\big]$ with additive error. In what follows, we assume that the input state is $\rho=|0^n\rangle\langle 0^n|$ and the desired outcome is $|0\rangle\langle 0|$. This assumption is without loss of generality, since we can reformulate both the state preparation and the measurement as quantum channels
\begin{equation}
	\mathcal{N}_1(\sigma)=\tr(\sigma)\rho,\qquad
	\mathcal{N}_L(\sigma)=\tr(E\sigma)|0\rangle\langle 0|+\tr((\mathds{1}-E)\sigma)|1\rangle\langle 1|.
\end{equation}
Consequently, we have
\begin{equation}
	\tr\big[|0\rangle\langle 0|(\mathcal{N}_L\circ\cdots\circ \mathcal{N}_1)(|0^n\rangle\langle 0^n|)\big]
	=\sum_{\overrightarrow{\bu}}W({|0\rangle\langle 0|}|\bu_L)\prod_{l=1}^{L}W_{\mathcal{N}_l}(\bu_l|\bu_{l-1})W_{|0^n\rangle\langle 0^n|}(\bu_0).
\end{equation}

To describe the simulation algorithm, we define the \emph{negativity} of quantum states and channels as
\begin{equation}
\begin{aligned}
	\cM_\rho&\coloneqq\|\rho\|_{W,1}=\sum_{\bu}|W_\rho(\bu)|,\\
	\cM_{\mathcal{N}}(\bu)&\coloneqq\|\cN(A_{\bu})\|_{W,1}=\sum_{\bu'}|W_{\mathcal{N}}(\bu'|\bu)|,\\
	\cM_{\mathcal{N}}&\coloneqq2^{\cM(\cN)}=\max_{\bu}\cM_\mathcal{N}(\bu).
\end{aligned}
\end{equation}
Then a noisy circuit comprised of the channels $\{\mathcal{N}_l\}_{l=1}^L$ can be simulated as follows. We sample the initial phase point $\bu_0$ according to the distribution $|W_{|0^n\rangle\langle 0^n|}(\bu_0)|/\cM_{|0^n\rangle\langle 0^n|}$ and, for each $l=1,\ldots,L$, we sample a phase point $\bu_l$ according to the conditional distribution $|W_{\mathcal{N}_l}(\bu_l|\bu_{l-1})|/\cM_{\mathcal{N}_l}(\bu_{l-1})$, after which we output the estimate
\begin{equation}
	\cM_{|0^n\rangle\langle 0^n|}\mathrm{Sign}\big[W_{|0^n\rangle\langle 0^n|}(\bu_0)\big]
	\prod_{l=1}^{L}\cM_{\mathcal{N}_l}(\bu_{l-1})\mathrm{Sign}\big[W_{\mathcal{N}_l}(\bu_l|\bu_{l-1})\big]
	W({|0\rangle\langle 0|}|\bu_L).
\end{equation}
This gives an unbiased estimate of the output probability since
\begin{equation}
\begin{aligned}
    &\mathbb{E}\Big[
    \cM_{|0^n\rangle\langle 0^n|}\mathrm{Sign}\big[W_{|0^n\rangle\langle 0^n|}(\bu_0)\big]
	\prod_{l=1}^{L}\cM_{\mathcal{N}_l}(\bu_{l-1})\mathrm{Sign}\big[W_{\mathcal{N}_l}(\bu_l|\bu_{l-1})\big]
	W({|0\rangle\langle 0|}|\bu_L)
    \Big]\\
    =&\sum_{\overrightarrow{\bu}}
    \frac{|W_{|0^n\rangle\langle 0^n|}(\bu_0)|}{\cM_{|0^n\rangle\langle 0^n|}}\prod_{l=1}^{L}\frac{|W_{\mathcal{N}_l}(\bu_l|\bu_{l-1})|}{\cM_{\mathcal{N}_l}(\bu_{l-1})}\\
    &\quad\times\cM_{|0^n\rangle\langle 0^n|}\mathrm{Sign}\big[W_{|0^n\rangle\langle 0^n|}(\bu_0)\big]
	\prod_{l=1}^{L}\cM_{\mathcal{N}_l}(\bu_{l-1})\mathrm{Sign}\big[W_{\mathcal{N}_l}(\bu_l|\bu_{l-1})\big]
	W({|0\rangle\langle 0|}|\bu_L)\\
	=&\sum_{\overrightarrow{\bu}}W({|0\rangle\langle 0|}|\bu_L)\prod_{l=1}^{L}W_{\mathcal{N}_l}(\bu_l|\bu_{l-1})W_{|0^n\rangle\langle 0^n|}(\bu_0)\\
	=&\tr\big[|0\rangle\langle 0|(\mathcal{N}_L\circ\cdots\circ \mathcal{N}_1)(|0^n\rangle\langle 0^n|)\big].
\end{aligned}
\end{equation}

Note that $\cM_{|0^n\rangle\langle 0^n|}=1$ since $|0^n\rangle$ is trivially a stabilizer state. Also for any stabilizer POVM $\{E_k\}$, we have $\tr\big[E_kA_{\bu}\big]\geq 0$ and
\begin{equation}
\sum_k\tr\big[E_kA_{\bu}\big]=\tr\big[A_{\bu}\big]=1,
\end{equation}
which implies that
\begin{equation}
\max_{\bu}|W({E_k}|\bu)|=\max_{\bu}\left|\tr\big[E_kA_{\bu}\big]\right|=\max_{\bu}\tr\big[E_kA_{\bu}\big]\leq 1.
\end{equation}
Therefore, the estimate that we output has absolute value bounded from above by
\begin{equation}
\cM_{\rightarrow}=
\prod_{l=1}^{L}\cM_{\mathcal{N}_l}.
\end{equation}
By the Hoeffding inequality, it suffices to take
\begin{equation}
	\frac{2}{\epsilon^2}\cM_{\rightarrow}^2\log\Big(\frac{2}{\delta}\Big)
\end{equation}
samples to estimate the probability of a fixed measurement outcome with accuracy $\epsilon$ and success probability $1-\delta$. 

In the description of the above algorithm, we have used the discrete Wigner representation of quantum states, channels, and measurement operators. However, the algorithm can be generalized using the frame and dual frame representation along the lines of the work \cite{PWB15}. Specifically, for any frame $\{F(\lambda):\lambda\in\Lambda\}$ and its dual frame $\{G(\lambda):\lambda\in\Lambda\}$ on a $d$-dimensional Hilbert space, we define the corresponding quasiprobability representation of a state, a channel, and a measurement operator respectively as
\begin{equation}
\begin{aligned}
	W_{\rho}(\lambda)&= \tr[F(\lambda)\rho],\\
	\cW_{\cN}(\lambda'|\lambda)&=\tr [F(\lambda')\cN(G(\lambda))],\\
	W(E|\lambda)&=\tr\!\big[EG(\lambda)\big].
\end{aligned}
\end{equation}
Then, we have similar rules for computing the measurement probability
\begin{equation}
	\sum_{\lambda,\lambda'}W(E|\lambda')W_{\mathcal{N}}(\lambda'|\lambda)W_\rho(\lambda)
	=\tr\big[E\mathcal{N}(\rho)\big]
\end{equation}
and the above discussion carries through without any essential change. For simplicity, we omit the details here. Note that for the discrete Wigner function representation that we used in our paper, the correspondence is $F(\lambda) = A_{\bu}/d$ and $G(\lambda) = A_\bu$.

\subsection{Comparison of classical simulation algorithms for noisy quantum circuits}

Recently, the channel robustness and the magic capacity were introduced  to quantify the magic of multi-qubit noisy circuits \cite{Seddon2019}. To be specific, given a quantum channel $\cN$, the channel robustness is defined as
\begin{align}\label{eq:rom decom N}
\cR_{*}(\cN)\coloneqq\min_{\Lambda_{\pm}\in\SCPO}\left\{2p+1: (1+p)\Lambda_+ - p\Lambda_- =\cN, p\ge0 \right\},
\end{align}
and  the magic capacity is defined as
\begin{align}\label{eq:cap N}
C(\cN)\coloneqq\max_{\ket \phi \in \STAB} \cR[(\operatorname{id}\ox \cN) (\proj \phi)].
\end{align}
They are related by the inequality \cite{Seddon2019}
\begin{align}
\label{eq:sandwich}
\cR(\Phi_\cN)\le C(\cN) \le \cR_*(\cN),
\end{align}
where $\cR(\cdot)$ is the robustness of magic (cf.~Section~\ref{sec: preliminary}) and $\Phi_\cN$ is the normalized \Choi operator of $\cN$. The authors of \cite{Seddon2019} further developed two matching simulation algorithms that scale quadratically with these channel measures. Here, we compare their approach with the one described in Section~\ref{sec: classical simulation mana} for simulating noisy qudit circuits. Note that neither the proof of \eqref{eq:sandwich} nor the static Monte Carlo algorithm of \cite{Seddon2019} depend on the dimensionality of the underlying system, so those results can be generalized to any qudit system with odd prime dimension.

Thus, we consider an $n$-qudit system with the underlying Hilbert space $\mathcal{H}_d^{\otimes n}$, where $d$ is an odd prime. Consider a noisy circuit consisting of the sequence  $\{\mathcal{N}_l\}_{l=1}^L$ of channels acting on the initial state $|0^n\rangle$, after which a computational basis measurement is performed. To describe the simulation algorithm based on channel robustness \cite{Seddon2019}, we assume each $\mathcal{N}_j$ has the optimal decomposition with respect to the set of CSPOs
\begin{equation}
	\mathcal{N}_j=(1+p_j)\mathcal{N}_{j,0}-p_j\mathcal{N}_{j,1},
\end{equation}
where $\cR_{*}(\cN_j)=2p_j+1$. For any $\vec{k}\in\mathbb{Z}_2^L$, define
\begin{equation}
	p_{\vec{k}}=\prod_{k_j=0}(1+p_j)\prod_{k_j=1}(-p_j)\qquad
	\|p\|_1=\sum_{\vec{k}}|p_{\vec{k}}|=\prod_j\cR_{*}(\cN_j)=\cR_{*}.
\end{equation}
We sample a $\vec{k}\in\mathbb{Z}_2^L$ from the distribution $|p_{\vec{k}}|/\|p\|_1$ and simulate the evolution $\mathcal{N}_{L,k_L}\circ\cdots\circ\mathcal{N}_{2,k_2}\circ\mathcal{N}_{1,k_1}$ \footnote{Strictly speaking, this evolution can be directly simulated only when the circuit is noiseless. Otherwise, we need to further decompose each $\mathcal{N}_{j,k_j}$ as a finite combination of stabilizer-preserving operations, each of which has a single Kraus operator. This does not affect the sample complexity of the simulator \cite{Seddon2019}.}. To achieve accuracy $\epsilon$ and success probability $1-\delta$, it suffices to take
\begin{equation}
	\frac{2}{\epsilon^2}\cR_{*}^2\log\Big(\frac{2}{\delta}\Big)
\end{equation}
samples.

To compare it with the mana-based simulation algorithm, we first prove that that the exponentiated mana of a quantum channel is always smaller than or equal to the channel robustness. To establish the separation, we introduce the robustness of magic with respect to non-negative Wigner function as follows.
\begin{definition}  
Given a quantum state $\rho$,  the robustness of magic with respect to non-negative Wigner function is defined as
\begin{align}\label{eq:rom decom}
\cR_{\cW_+} (\rho)&\coloneqq \min\left\{2p+1: \rho =(1+p)\sigma -p\omega,  \omega,\sigma \in \cW_+\right\}.
\end{align}
\end{definition}
Since $\STAB\subset \cW_+$, we have
\begin{align}\label{eq:rom ineq}
\cR_{\cW_+} (\Phi_\cN) \le \cR(\Phi_\cN) \le C(\cN) \le \cR_*(\cN).
\end{align}

\begin{proposition}
\label{prop:rom_mana gap}
Given a quantum channel $\cN$, the following inequality holds
\begin{equation}
2^{\cM(\cN)}=\cM_{\cN} \le \cR_{*}(\cN),
\end{equation}
and the inequality can be strict.
\end{proposition}
\begin{proof}
Suppose $\{p,\Lambda_{\pm}\}$ is the optimal solution to Eq.~\eqref{eq:rom decom N} of $\cR_*(\cN)$.

Then, we have
\begin{align}
	\cM_{\cN}&=\max_{\bu} \|\cN(A_\bu)\|_{W,1}\\
	&=\max_{\bu} \|(1+p)\Lambda_+(A_\bu)-p\Lambda_-(A_\bu)\|_{W,1}\\
	&\le \max_{\bu}\left[(1+p) \|\Lambda_+(A_\bu)\|_{W,1}+p\|\Lambda_-(A_\bu)\|_{W,1}\right]\label{eq: M R triangle}\\
	&=2p+1 \label{eq: M R final}\\
	&=\cR_*(\cN).
\end{align}
The inequality in \eqref{eq: M R triangle} follows due to the triangle inequality.  The equality in \eqref{eq: M R final} follows since $\Lambda_{\pm}\in \SCPO$ and then $\|\Lambda_\pm(A_\bu)\|_{W,1}=1$ for any $\bu$.

Furthermore, we demonstrate the strict separation between $2^{\cM(\cN)}$ and $\cR_{*}(\cN)$ via the following example. Let us consider the diagonal unitary
\begin{align}
U_{\theta}=\left(\begin{matrix}
e^{i\theta/9} & 0 & 0\\
0 & 1 & 0\\
0& 0 & e^{-i\theta/9}
\end{matrix}
\right).
\end{align}
Note that the  $T$ gate is a special case, given by $U_{2\pi}$. Due to Eq.~\eqref{eq:rom ineq}, the separation between  $\cM({U_\theta})$ and $\log \cR_{\cW_+} (\Phi_{U_\theta})$  in Figure~\ref{fig: comp R M} indicates that the mana of a channel can be strictly smaller than the channel robustness and magic capacity, i.e.,
\begin{align}\label{eq: gap mana rom}
\cM_{U_\theta} < \cR_{\cW_+} (\Phi_{U_\theta}) \le C({U_\theta})\le R_{*}({U_\theta}).
\end{align}
This concludes the proof.
\end{proof}

\bigskip

\begin{figure}
\begin{center}
\includegraphics[width=8.7cm]{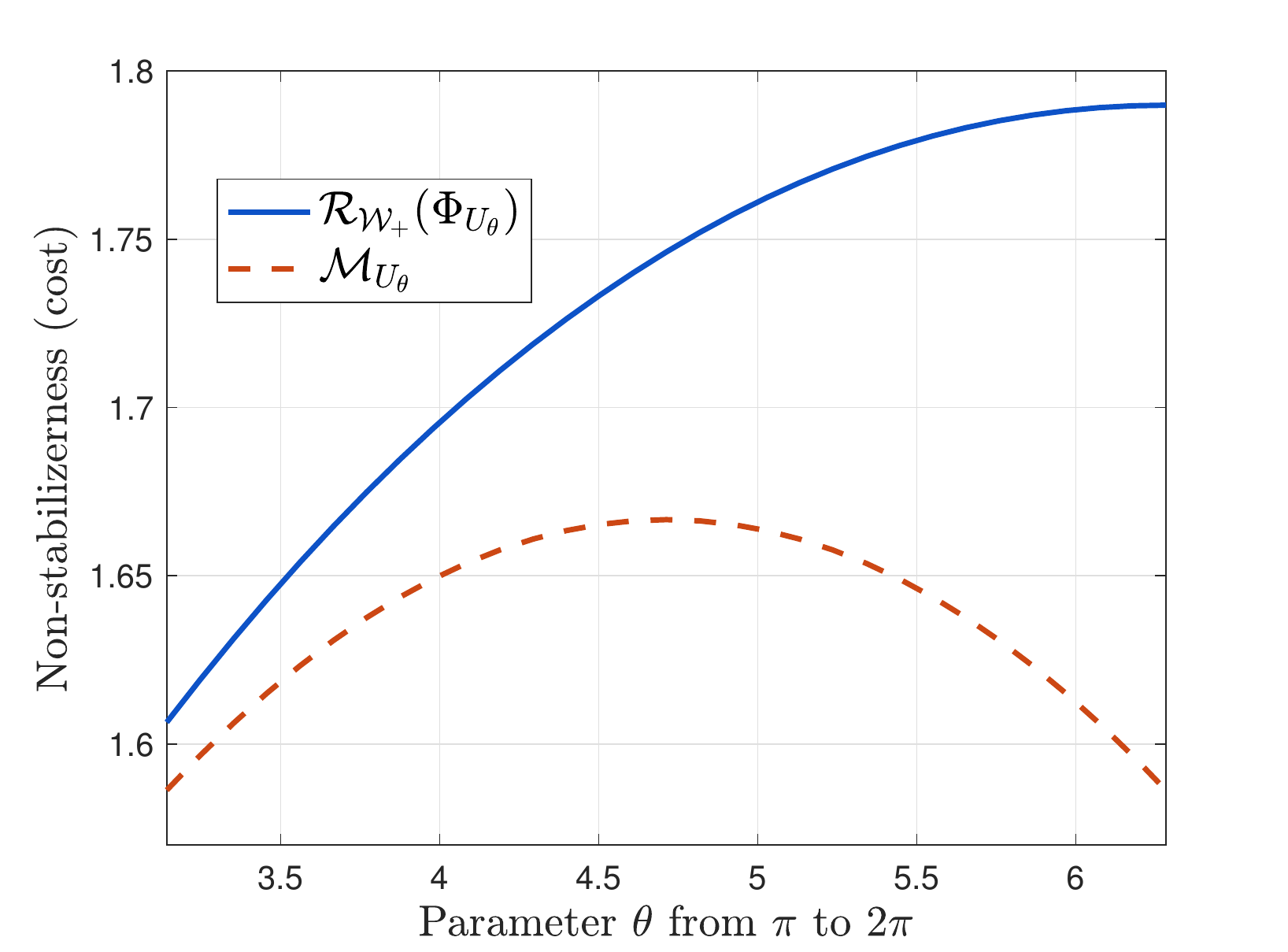}
\caption{Comparison between $\cM_{U_\theta}$ and $ \cR_{\cW_+} (\Phi_{U_\theta})$ for $\pi\le \theta \le 2\pi$. The gap indicates that $\cM_{U_\theta}$ is strictly smaller than $C({U_\theta})$ and $R_{*}({U_\theta})$.}
\label{fig: comp R M}
\end{center}
\end{figure}

Applying Proposition~\ref{prop:rom_mana gap} to the channels $\{\mathcal{N}_l\}_{l=1}^L$, we find that
\begin{equation}
	\cM_{\rightarrow}=
	\prod_{j}\cM_{\mathcal{N}_j}
	\leq\prod_j\cR_{*}(\cN_j)=\cR_{*}.
\end{equation}
Thus for an $n$-qudit system with odd prime dimension, the sample complexity of the mana-based approach is never worse than the algorithm of \cite{Seddon2019} based on channel robustness. Furthermore, the separation demonstrated in \eqref{eq: gap mana rom} indicates that the mana-based algorithm can be strictly faster for certain quantum circuits. 

The above example shows that the mana of a quantum channel can be smaller than its magic capacity \cite{Seddon2019}, due to \eqref{eq: gap mana rom}, but it is not clear whether this relation holds for general quantum channels.

 %%%%%%%%%%%%%%%%
 \section{Examples}

\subsection{Non-stabilizerness under depolarizing noise}
For near term quantum technologies, certain physical noise may occur during quantum information processing. One common quantum noise model is given by the depolarizing channel:
\begin{align}
{\cD_p}(\rho) = (1-p) \rho + \frac{p}{d^2-1} \sum_{\substack{0\leq i,j\leq d-1 \\ (i,j)\neq (0,0)}} X^i Z^j \rho (X^i Z^j)^\dagger,
\end{align} 
where $p\in[0,1]$ and $X,Z$ are the generalized Pauli operators. 
 
Let us suppose that depolarizing noise occurs after the implementation of a $T$ gate. From Figure~\ref{fig: noise T}, we find that if the depolarizing noise parameter $p$ is higher than or equal to $0.62$, then the channel $\cD_p \circ \cT$ cannot generate any non-stabilizerness. That is, the channel $\cD_p \circ \cT$ becomes CPWP after this cutoff.
\begin{figure}
\begin{center}
\includegraphics[width=8.2cm]{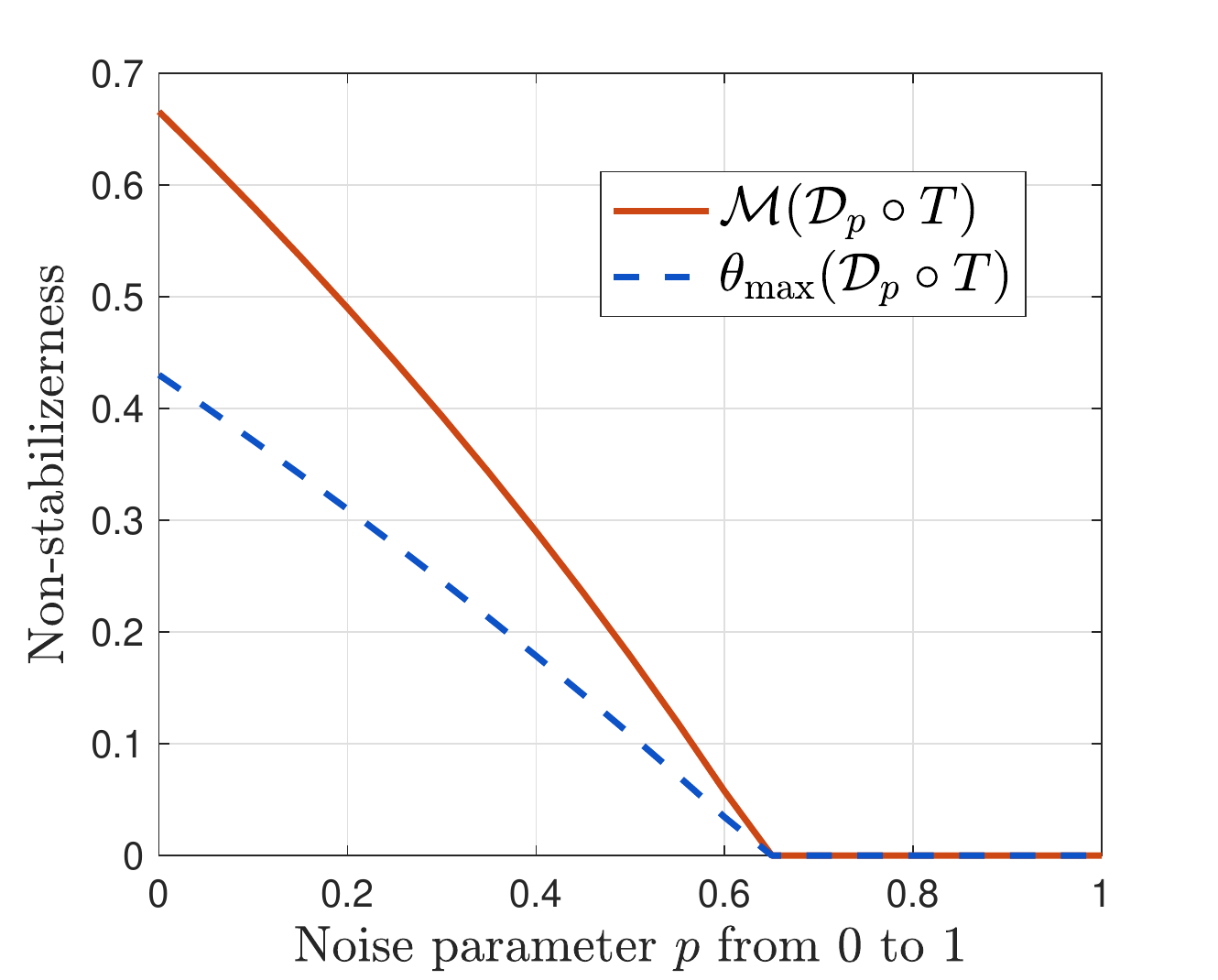}
\caption{Non-stabilizerness of $T$ gate after depolarizing noise. The solid red line quantifies the classical simulation cost of noisy circuits $\cD_p\circ \cT$. The dashed blue line gives the upper bound of magic generating capacity of $\cD_p\circ \cT$.}
\label{fig: noise T}
\end{center}
\end{figure}

Another interesting case is the CCX gate. Let us suppose that depolarizing noise occurs in parallel after the implementation of the CCX gate. The mana of $\cD_p^{\ox 3} \circ CCX$ is plotted in Figure~\ref{fig: noise CCX}, where we see that it decreases linearly and becomes equal to zero at around $p \approx 0.75$.
\begin{figure}
\begin{center}
\includegraphics[width=8.5cm]{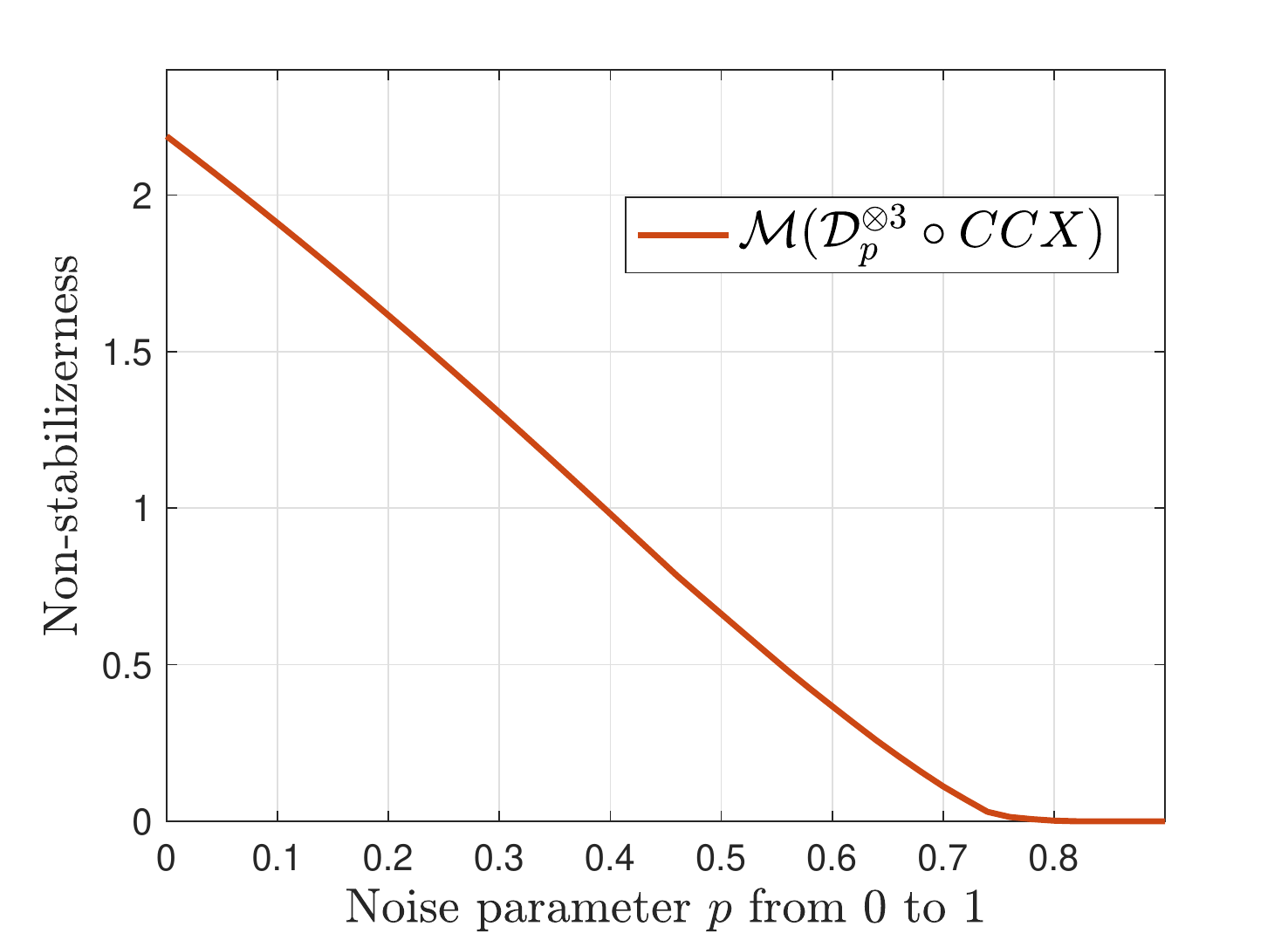}
\caption{Non-stabilizerness of CCX gate after depolarizing noise. The solid line also quantifies the classical simulation cost of the noisy circuit $\cD_p^{\ox 3}\circ CCX$.}
\label{fig: noise CCX}
\end{center}
\end{figure}

\subsection{Werner-Holevo channel}

An interesting qutrit channel is the qutrit Werner-Holevo channel \cite{WH02}:
\begin{align}
\cN_{\operatorname{WH}}(V)=\frac{1}{2}[(\tr V) \1- V^T].
\end{align}
In what follows, we find that the Werner-Holevo channel maps any quantum state to a free state in $\cW_+$ (state with non-negative Wigner function), while its amortized magic is given by our channel measure. This also indicates that the ancillary reference system is necessary to consider in the study of the resource theory of magic channels.
\begin{proposition}
For the qutrit Werner-Holevo channel $\cN_{\operatorname{WH}}$,
\begin{align}
    \cN_{\operatorname{WH}}(\rho)\in \cW_+,
\end{align}
for any input state $\rho$ (which is restricted to be a state of the channel input system),
while
\begin{align}
    \cM^{\cA}(\cN_{\operatorname{WH}}) &= \cM(\cN_{\operatorname{WH}})=\frac{5}{3},\\
    \theta_{\max}^{\cA}(\cN_{\operatorname{WH}}) & = \theta_{\max}(\cN_{\operatorname{WH}})=\frac{5}{3}.
\end{align}
\end{proposition}
\begin{proof}
On the one hand, for any input state $\rho$, we find that
\begin{align}
W_{\cN_{\operatorname{WH}}(\rho)}(\bu)  =\frac{1}{6}\tr A_{\bu}(\1-\rho^T) =\frac{1}{6}(1-\tr A_{\bu}\rho^T)
 \geq \frac{1}{6}(1-\| A_{\bu}\|_\infty)\ge 0,
\end{align}
where the first inequality follows because $\|A_{\bu}\|_{\infty} = \|A_{0}\|_{\infty} = 1$, since $A_{\bu} = T_\bu A_0 T_\bu^\dag$, the matrix $T_\bu$ is unitary, and $A_0$ for a qutrit is explicitly given by the following unitary transformation:
\begin{equation}
A_0 = \begin{bmatrix}
1 & 0 & 0 \\
0 & 0 & 1\\
0 & 1 & 0
\end{bmatrix}.
\end{equation}

On the other hand, we can set $\rho_{RA}$ to be the maximally entangled state, and we find from numerical calculations that
\begin{align}
    \cM((\operatorname{id}_R \ox \cN_{\operatorname{WH}})(\rho_{RA}))-\cM(\rho_{RA}) = \frac{5}{3}.
\end{align}
Meanwhile, we find from numerical calculations that $\cM(\cN_{\operatorname{WH}})=5/3$, which by Proposition~\ref{prop: M amo} means that 
\begin{align}
    \sup_{\rho_{RA}} [\cM((\operatorname{id}_R \ox \cN_{\operatorname{WH}})(\rho_{RA}))-\cM(\rho_{RA})] = \cM(\cN_{\operatorname{WH}}) = \frac{5}{3}.
\end{align}

Similarly, we find from numerical calculations that
\begin{align}
    \sup_{\rho_{RA}} [\theta_{\max}((\operatorname{id}_R \ox \cN_{\operatorname{WH}})(\rho_{RA}))-\theta_{\max}(\rho_{RA})] = \theta_{\max}(\cN_{\operatorname{WH}})=\frac{5}{3}.
\end{align}
This concludes the proof.
\end{proof}

 \section{Conclusion}
 We have introduced two efficiently computable magic measures of quantum channels to quantify and characterize the non-stabilizer resource possessed by quantum channels. These two channel measures have application in evaluating magic generating capability, gate synthesis, and classical simulation of noisy quantum circuits. More generally, our work establishes fundamental limitations on the processing of quantum magic using noisy quantum circuits, opening new perspectives for the investigation of the resource theory of quantum channels in fault-tolerant quantum computation. 

One future direction is to explore tighter evaluations of the distillable magic of quantum channels. We think that it would also be interesting to explore other applications of our channel measures and generalize our approach to the multi-qubit case.

\section*{Acknowledgements}
We are grateful to the anonymous referees for several comments that helped to improve our paper.
XW acknowledges support from the Department of Defense. MMW acknowledges support from NSF under grant no.~1714215. YS acknowledges support from the Army Research Office (MURI award W911NF-16-1-0349), the Canadian Institute for Advanced Research, the National Science Foundation (grants 1526380 and 1813814), and the U.S.\ Department of Energy, Office of Science, Office of Advanced Scientific Computing Research, Quantum Algorithms Teams and Quantum Testbed Pathfinder programs.

\appendix

\section{On completely positive maps with non-positive mana}

\label{app:tni-lemma}

\begin{lemma}
\label{lem:mana-CP-maps}
Let $\cE$ be a completely positive map. If $\cM(\cE) \leq 0$, then $\cE$ is trace non-increasing on the set~$\cW_+$ (quantum states with non-negative Wigner function).
\end{lemma}
\begin{proof}
The assumption $\cM(\cE) \leq 0$ is equivalent to $\max_{\bu} \sum_{\bv}  |W_{\cE}(\bv|\bu)| \leq 1$. For  $\rho \in \cW_+$, then consider that
\begin{align}
\tr[\cE(\rho)] & = \left| \tr[\cE(\rho)] \right|   = \left| \sum_{\bv, \bu} W_{\cE}(\bv|\bu) W_{\rho}(\bu)\right| \\
& \leq  \sum_{\bv, \bu} \left| W_{\cE}(\bv|\bu)\right| \left| W_{\rho}(\bu)\right|  \leq  \sum_{\bu}  \left| W_{\rho}(\bu)\right| = \sum_{\bu}   W_{\rho}(\bu) = 1.
\end{align}
The first equality follows from the assumption that $\cE$ is completely positive and $\rho$ is positive semi-definite. The second equality follows from Lemma~\ref{lem:prop-dWF}. The first inequality follows from the triangle inequality and the second from the assumption $\cM(\cE) \leq 0$. The final two equalities follow from the assumption $\rho \in \cW_+$ and  \eqref{eq:dWF-normalized}.
\end{proof}

\section{Data processing inequality for the Wigner trace norm}

\label{app:DP-Wig-t-norm}

\begin{lemma}
\label{lem:DP-Wigner-tnorm}
For any operator $Q$ and CPWP channel $\Pi$, the following inequality holds
\begin{align}
    \|\Pi(Q)\|_{W,1} \le \|Q\|_{W,1}.
\end{align}
\end{lemma}
\begin{proof}
Let us suppose that $Q=\sum_\bv q_\bv A_\bv$. Then we have $\|Q\|_{W,1}=\sum_\bv |q_\bv|$. Furthermore,
\begin{align}
    \|\Pi(Q)\|_{W,1} & = \left\|\sum_\bv q_\bv \Pi(A_\bv)\right\|_{W,1} \\
    & = \sum_\bu \left|\sum_\bv q_\bv\tr [A_\bu\Pi(A_\bv)]\right| \\
    & = \sum_\bu \left|\sum_\bv q_\bv W_{\Pi}(\bu|\bv)\right|\\
    & \le \sum_\bu \sum_\bv |q_\bv| W_{\Pi}(\bu|\bv)\\
    & =\sum_\bv |q_\bv| 
     = \|Q\|_{W,1}.
\end{align}
This concludes the proof.
\end{proof}

\bibliographystyle{myhamsplain2}
\bibliography{smallbib}

%\appendix
 \end{document}